\documentclass[12pt]{article} \usepackage{amssymb}
\usepackage{amsfonts} \usepackage{amsmath} \usepackage{amsthm} \usepackage{bm}
\usepackage{latexsym} \usepackage{epsfig}
\usepackage{hyperref}
\usepackage{mathrsfs}

\newtheorem*{theorem*}{Theorem}
\newtheorem{theorem}{Theorem}[section]
\newtheorem{lemma}[theorem]{Lemma}
\newtheorem*{proposition*}{Proposition}
\newtheorem{maintheorem}{Theorem}

\newtheorem{claim}[theorem]{Claim}

\newtheorem{corollary}[theorem]{Corollary}

\newtheorem{definition}[theorem]{Definition}

\newcommand{\ignore}[1]{}

\newcommand{\enote}[1]{} \newcommand{\knote}[1]{}
\newcommand{\rnote}[1]{}

\newcommand{\argmax}{\operatornamewithlimits{argmax}}

\newcommand{\cD}{{\mathcal{B}}}
\newcommand{\cB}{{\mathcal{B}}}
\newcommand{\cR}{{\mathcal{R}}}
\newcommand{\cT}{{\mathcal{T}}}
\newcommand{\cF}{{\mathcal{F}}}

\newcommand{\cG}{{\mathcal{G}}}

\newcommand{\cK}{{\mathcal{K}}}

\renewcommand{\P}[1]{{\mathbb{P}}\left[{#1}\right]}
\newcommand{\CondP}[2]{{\mathbb{P}}\left[{#1}\middle\vert{#2}\right]}
\newcommand{\E}[1]{{\mathbb{E}}\left[{#1}\right]}
\newcommand{\ind}[1]{{\bf 1}_{#1}}

\newcommand{\CondE}[2]{{\mathbb{E}}\left[{#1}\middle\vert{#2}\right]}

\newcommand{\eps}{\epsilon}

\newcommand{\N}{\mathbb N} \newcommand{\R}{\mathbb R}

\newcommand{\half}{{\textstyle \frac12}}

\renewcommand{\phi}{\varphi}

\newcommand{\belief}{I}

\newcommand{\psignal}{W} %private signal
\newcommand{\action}{A}

\newcommand{\dtv}{d_{TV}}

\newcommand{\strat}{Q}
\newcommand{\stratp}{\bar{\strat}}

\newcommand{\stratx}{R}
\newcommand{\stratpx}{\bar{\stratx}}

\newcommand{\util}{u}
\newcommand{\neigh}[1]{N(#1)}
\newcommand{\hist}[2]{\action^{\neigh{#1}}_{[0,#2)}}
\newcommand{\disc}{\lambda}
\newcommand{\optset}{C}
\newcommand{\bestres}{B}
\newcommand{\resp}{r}
\newcommand{\dist}{\Delta}

\newcommand{\gs}{\mathcal{GS}}
\newcommand{\info}{\mathcal{F}}
\newcommand{\learnprob}{p}
\newcommand{\mapS}{\hat{S}}
\newcommand{\estest}{R}

\newcommand{\infdep}{\mathrm{dep}}
\newcommand{\pstar}{p^*(\mu_0,\mu_1)}
\newcommand{\responseSpace}{\mathscr{H}}
\newcommand{\responseSubset}{\mathcal{H}}
\newcommand{\cJ}{\mathcal{J}}
\newcommand{\cM}{\mathcal{M}}
\newcommand{\graphs}{\mathcal{E}}
\newcommand{\scg}{\mathcal{SCG}}
\newcommand{\eqs}{\mathcal{EQ}}

\begin{document}
\title{Strategic Learning and the Topology of Social Networks}

\author{Elchanan Mossel\footnote{University of Pennsylvania and
    University of California, Berkeley. E-mail:
    mossel@stat.berkeley.edu. Supported by NSF award DMS 1106999, by
    ONR award N000141110140 and by ISF grant 1300/08.}~, Allan
  Sly\footnote{University of California, Berkeley. Supported by a
    Sloan Research Fellowship in mathematics and by NSF award DMS
    1208339} and Omer Tamuz\footnote{Massachusetts Institute of
    Technology and Microsoft Research New England. This research was
    supported in part by a Google Europe Fellowship.}}

\maketitle
\begin{abstract} 
  We consider a group of strategic agents who must each repeatedly
  take one of two possible actions. They learn which of the two
  actions is preferable from initial private signals, and by observing
  the actions of their neighbors in a social network.

  % We consider a group of strategic agents who learn which of two
  % possible actions is preferable, from conditionally independent
  % private signals, by repeatedly observing the actions of their
  % neighbors in a social network.

  We show that the question of whether or not the agents learn
  efficiently depends on the topology of the social network. In
  particular, we identify a geometric ``egalitarianism'' condition on
  the social network that guarantees learning in infinite networks, or
  learning with high probability in large finite networks, in any
  equilibrium. We also give examples of non-egalitarian networks with
  equilibria in which learning fails.
  

\end{abstract}

\section{Introduction}

Consider a group in which each agent faces a repeated choice between
two actions. Initially, the information available to each agent is a
private signal, which gives a noisy indication of which is the correct
action. As time progresses, the agents learn more by observing the
actions of their neighbors in a social network. They do not, however,
obtain any direct indication of the payoffs from their actions.  For
example, their choice could be one of lifestyle, where one can learn
by observing the actions of others, but where payoffs (e.g.,
longevity) are only revealed after a large amount of
time\footnote{Consider parents who, each night, decide whether to lay
  their baby to sleep on its back or on its stomach. They can learn by
  observing the actions of their peers, but presumably do not receive
  any direct feedback regarding the effect of their actions on the
  baby's health.}.

We are interested in the question of {\em learning}, or {\em
  aggregation of information}: When is it the case that, through
observing each other, the agents exchange enough information to
converge to the correct action? In particular, we are interested in
the role that the geometry of the social network plays in this
process, and in its effect on learning. Which social networks enable
the flow of information, and which impede it?  This problem has been
studied extensively in the literature, using mostly boundedly-rational
or heuristic approaches~\cite{DeGroot:74, EllFud:95, BalaGoyal:96,
  DeMarzo:03, golub2010naive, jadbabaie2013information}. However, the
basic question of how {\em strategic} agents behave in this setting
has been largely ignored\footnote{Notable exceptions
  are~\cite{mossel2012asymptotic} and~\cite{arieli2013inferring}; we
  discuss these below.}, perhaps because the model is mathematically
difficult to approach, or because strategic behavior seems
unfeasible\footnote{See, e.g., Bala and Goyal~\cite{BalaGoyal:96}:
  {\em ``to keep the model mathematically tractable... this
    possibility [strategic agents] is precluded in our
    model... simplifying the belief revision process
    considerably.''}}. This article aims to fill this gap. We define a
notion of egalitarianism for social networks, and show that when agents
are strategic, learning always occurs on egalitarian social networks,
and may not occur on those that are not egalitarian. Interestingly,
these results broadly resemble those of some of the heuristic models
(see, e.g., Golub and Jackson~\cite[Theorem 1]{golub2010naive}).

We call a social network graph {\em $(d,L)$-egalitarian} if it
satisfies the following two conditions: (1) At most $d$ edges leave
each node (that is, each agent observes at most $d$ others), and (2)
whenever there is an edge from node $i$ to $j$, there is a path {\em
  back} from $j$ to $i$, of length at most $L$ (that is, no agent is
too far removed from those who observe her). In this article we show
that on connected $(d,L)$-egalitarian graphs the agents learn the
correct action, and give examples of non-egalitarian graphs in which
learning fails.

Our model is a discounted, repeated game with incomplete
information. We consider a state of nature $S$ which is equal to
either $0$ or $1$, with equal probability. Each agent receives a
private signal that is independent and identically distributed
conditioned on $S$, and is correlated with $S$. In each discrete time
period $t$, each agent $i$ chooses an action $\action^i_t$ taking
values in $\{0,1\}$. The information available to her is her own
private signal, as well as the actions of her social network neighbors
in the previous time periods. Agent $i$'s stage utility at time period
$t$ is equal to $1$ if $\action^i_t=S$ and to $0$ otherwise, and is
discounted exponentially, by a common rate. We consider general Nash
equilibria, and show that they indeed exist
(Theorem~\ref{thm:equi-exists}); this does not follow from standard
results.

We say that agent $i$ {\em learns} $S$ when $\action^i_t$ is equal to
$S$ from some time on, and that {\em learning} takes place when all
agents learn $S$. Our main result
(Theorem~\ref{thm:learning-finite-intro}) is that on connected
$(d,L)$-egalitarian graphs, in any equilibrium, learning occurs with
high probability on large graphs, and with probability one on infinite
graphs. We do not impose unbounded likelihood ratios: learning occurs
in egalitarian networks even for weak - but informative - signals
(contrast this with the sequential learning case of Smith and
S{\o}rensen~\cite{smith2000pathological}). Note that this applies to
all Nash equilibria, and therefore in particular to any perfect
Bayesian equilibria.  We also provide examples of equilibria on large
non-egalitarian networks in which, with non-vanishing probability, the
agents do not learn.

Our results require a smoothness condition on private signals: each
private belief (the probability that $S=1$, conditioned on the private
signal) must have a non-atomic distribution. This ensures that agents
are never (i.e., with zero probability) indifferent.  Our results do
not, in general, hold without this condition; indifference can impede
the flow of information (see, e.g.,~\cite[Example
A.1]{mossel2012asymptotic}). While real life signals are arguably
always discrete or even finite, we propose that even with this
requirement it is still possible to model or approximate a large range
of signals.

The model we study makes heavy demands on the agents in terms of
rationality, common knowledge, and human computation: agents are
assumed to maximize a complicated expected utility function, to know
the structure of the entire social network, and to precisely make
complicated inferences regarding the state of nature. While our
approach is standard in this literature (see,
e.g.,~\cite{GaleKariv:03, rosenberg2009informational,
  AcemMuntDahlLobelOzd:08}), these features of our model prompt us to
present our results as benchmarks, rather than as predictive
statements about the world.

%The condition of non-atomic private beliefs ensures that agents are
%never (i.e., with probability zero) indifferent in their choice
%between the two actions.   Under this condition, we prove in
%Theorem~\ref{thm:agreement} that the agents agree on their limiting
%actions; this is an important ingredient of the proof of
%Theorem~\ref{thm:learning-finite-intro}. This agreement result
%elucidates an important theorem of Rosenberg, Solan and
%Vieille~\cite{rosenberg2009informational}, who consider the question
%of when agents eventually agree (regardless of whether or not they
%learn) and show that agents (in this game, among a larger class) can
%disagree only if they are indifferent.

\vspace{0.25in}

The rest of this article proceeds as follows. In
section~\ref{sec:illustrative} we discuss an example of a
$(d,L)$-egalitarian graph, using it to provide intuition into the
ideas behind our main result
(Theorem~\ref{thm:learning-finite-intro}). In
Section~\ref{sec:non-learning} we provide two examples of
non-egalitarian graphs on which the agents fail to learn. In
Section~\ref{sec:main-defs} we introduce our model formally. In
Section~\ref{sec:agreement} we explore the question of agreement and
show that indeed the agents all converge to the same action.
Section~\ref{sec:topologies} includes our main technical contribution:
a topology on equilibria of this game, as seen from the point of view
of a particular agent. In Section~\ref{sec:learning} we prove
Theorem~\ref{thm:learning-finite-intro}, and
Section~\ref{sec:conclusion} provides a conclusion.

\subsection{Related literature}
Learning on social networks is a widely studied field; a complete
overview is beyond the scope of this paper, and so we shall note only
a few related studies.

Bala and Goyal~\cite{BalaGoyal:96} study a similar model, and show
results of learning or non-learning in different cases. Their model is
boundedly-rational, with agents not taking into account the choices of
their neighbors when forming their beliefs. Other notable bounded
rationality models of learning through repeated social interaction are
those of DeGroot~\cite{DeGroot:74}, Ellison and
Fudenberg~\cite{EllFud:95}, DeMarzo, Vayanos and
Zwiebel~\cite{DeMarzo:03}, Golub and Jackson~\cite{golub2010naive} and
recently Jadbabaie, Molavi and
Tahbaz-Salehi~\cite{jadbabaie2013information}. Interestingly, a
recurring theme is that learning is facilitated by graphs which are
egalitarian, although notions of egalitarianism differ across models
(see, e.g., Golub and Jackson~\cite[Property 2]{golub2010naive}).

In a previous paper~\cite{mossel2012asymptotic}, we consider the same
question, but for myopic agents. The analysis in that case is far
simpler and does not require the technical machinery that we construct
in this article. More importantly, the conditions for learning are
qualitatively different for myopic agents, as compared to those for
strategic agents: in the myopic setting, the upper bound on the number
of observed neighbors is not needed. In fact, myopic agents learn with
high probability on networks with no uniform upper bound. Thus there
are examples of graphs on which myopic agents learn but strategic
agents do not. We elaborate on this in our second example of
non-learning, in Section~\ref{sec:non-learning}.

In concurrent work by Arieli and
Mueller-Frank~\cite{arieli2013inferring}, learning results are derived
in a strategic setting with richer actions spaces; they study models
in which actions are rich enough to reveal beliefs, and show that in
that case learning occurs under general conditions, and in particular
for any graph topology.  To the best of our knowledge, no previous
work considers learning, in repeated interaction, on social networks,
in a fully rational, strategic setting.

The study of {\em agreement} (rather than learning) on social networks
is also related to our work, and in fact we make crucial use of the
work of Rosenberg, Solan and
Vieille~\cite{rosenberg2009informational}, who prove an agreement
result for a large class of games with informational externalities
played on social networks.  This is a field of study founded by
Aumann's ``Agreeing to disagree'' paper~\cite{aumann1976agreeing}, and
elaborated on by Sebenius and Geanakoplos~\cite{sebenius1983don},
McKelvey and Page~\cite{mckelvey1986common}, Parikh and
Krasucki~\cite{parikh1990communication}, Gale and
Kariv~\cite{GaleKariv:03}, M{\'e}nager~\cite{menager2006consensus} and
recently Mueller-Frank~\cite{mueller2010general}, to name a few. The
moral of this research is that, by-and-large, rational agents
eventually reach consensus, even in strategic settings.  We elaborate
on the work of Rosenberg, Solan and
Vieille~\cite{rosenberg2009informational} and show that when private
signals are non-atomic then, asymptotically, agents agree on best
responses (Theorem~\ref{thm:agreement}). This agreement result is an
important ingredient of our main learning result
(Theorem~\ref{thm:learning-finite-intro}).

Another strain of related literature is that of {\em herd behavior},
started by Banerjee~\cite{Banerjee:92} and Bikhchandani, Hirshleifer
and Welch~\cite{BichHirshWelch:92}, with significant generalizations
and further analysis by Smith and
S{\o}rensen~\cite{smith2000pathological}, Acemoglu, Dahleh, Lobel and
Ozdaglar~\cite{AcemMuntDahlLobelOzd:08} and recently Lobel and
Sadler~\cite{lobel2012social}. Here, the state of nature and private
signals are as in our model, and agents are rational. However, in
these models agents act sequentially rather than repeatedly. The same
informational framework is also shared by models of committee behavior
and committee mechanism design (cf.\ Laslier and
Weibull~\cite{laslier2008committee}, Glazer and
Rubinstein~\cite{glazer1998motives}).

\subsection{Acknowledgments}
We would like to thank Shahar Kariv for introducing us to this
field. For commenting on drafts of this paper we would like to thank
Nageeb Ali, Ben Golub, Eva Lyubich, Markus Mobius, Ariel Rubinstein,
Ran Shorrer, Glen Weyl, and especially Scott Kominers.

\section{An illustrative example}
\label{sec:illustrative}

\begin{figure}[h]
  \centering
  \includegraphics[scale=0.8]{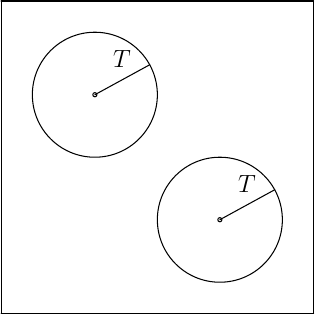}
  \caption{\label{fig:z2} Learning in symmetric equilibria
    on the two dimensional grid.}
\end{figure}

To provide some intuition for why agents learn on egalitarian graphs
(Theorem~\ref{thm:learning-finite-intro}) we consider the simple,
particular case that the graph is the undirected, infinite two
dimensional grid, in which each agent has four neighbors. This is a
$(4,1)$-egalitarian graph, and so
Theorem~\ref{thm:learning-finite-intro} says that the agents learn
$S$, or that, equivalently, in any equilibrium the actions of every
agent converge to $S$. We now explain why this is indeed the case,
under the further restriction to symmetric equilibria.

The first step in proving that all agents converge to $S$ is to show
that all agents converge to the same action, which we do in
Theorem~\ref{thm:agreement}. This result uses - and perhaps elucidates
- an important theorem of Rosenberg, Solan and
Vieille~\cite{rosenberg2009informational}, who consider the question
of when agents eventually agree, regardless of whether or not they
learn.  For a large class of games which includes the one we consider,
they show that agents can disagree only if they are indifferent. Our
additional requirement of non-atomic private signals allows us to rule
out the possibility of indifference, and show that all agents converge
to the same action\footnote{In fact, Theorem~\ref{thm:agreement} does
  not exclude the case that no agent converges at all; we will, for
  now, ignore this possibility.}.

Having established that all agents converge to the same action, we use
the fact that the graph is symmetric, as is the equilibrium. Hence all
agents converge at the same ex-ante rate, and therefore, at some large
enough time $T$, any particular agent will have converged, except with
some very small probability $\eps$. Of course, since the graph is
infinite, there will be at time $T$ many agents who have yet to
converge. However, if we consider any one agent (or two, as we do
immediately below), the probability of non-convergence is
negligible.

Now, consider two agents which are more than $2T$ edges apart on the
graph (see Figure~\ref{fig:z2}), and condition on the state of nature
$S$ equaling one. The two agents' actions at time $T$ are independent
random variables (conditioned on the state of nature), as they are too
far apart for any information to have been exchanged between them. On
the other hand, since all agents converge to the same action, these
independent random variables are equal (except with probability $\sim
2\eps$); the two agents somehow, with high probability, reach the same
conclusion independently.

Now, two independent random variables that are equal must be
constant. The agents' actions at time $T$ are equal with high
probability and independent conditioned on $S$, and so are with high
probability equal to some fixed action. Since the agents' signals are
informative, this action is more likely to equal the state of nature
than not (Claim~\ref{clm:p-half}). Since this holds for every
$\eps>0$, every agent's limit action must equal the state of nature.

\subsection{General egalitarian graphs}
The formalization and extension of this intuition to general
egalitarian graphs and general (i.e., non-symmetric) equilibria
requires a significant technical effort, and in fact the construction
of novel tools for the analysis of games on networks; to this we
devote most of the rest of this article. We now provide an overview of
the main ideas.

The main notion we use is one of {\em compactness}. The two
dimensional grid graph ``looks the same'' from the point of view of
every node: there is only one ``point of view'' in this graph. Such
graphs as known as {\em transitive graphs} in the mathematics
literature. Note that this is the only property of the grid that we
used in the proof sketch above, and therefore the same idea can be
applied to all symmetric equilibria on infinite, connected transitive
graphs. 

We formalize a notion of an ``approximate points of view''. We show
that in particular, in an egalitarian graph, the nodes of the graph
can be grouped into a finite number of sets, where from each set the
graph ``looks {\em approximately} the same''. Formally, we construct a
topology in which the set of points of view in a graph is precompact
if and only if the graph is $(d,L)$-egalitarian for some $d$ and $L$
(Theorem~\ref{thm:compact-graph}). In this sense, egalitarianism, in
which the set of points of view is precompact, is a relaxation of
transitivity, in which the set of points of view is a
singleton. Indeed, transitivity is an extreme notion of
egalitarianism, by any reasonable definition of an egalitarian graph.

This property of egalitarian graphs allows us to apply the intuition
of the above example (or, more precisely, a similar intuition) to any
infinite, $(d,L)$-egalitarian graph
(Theorem~\ref{thm:compact-learning}).  The fact that general
equilibria are not symmetric is similarly treated by establishing that
the space of equilibria is compact (Claim~\ref{clm:cB-compact}). The
theorem on finite graphs is proved by reduction to the case of
infinite graphs.

Our main technical innovation is the construction of a topology on
equilibria of this game, as seen from the point of view of a
particular agent
(Section~\ref{sec:rooted-graph-strategy-profiles}). In this topology,
an equilibrium has a finite number of ``approximate points of view''
if and only if the graph is egalitarian (Claim~\ref{thm:conv-conv}).
This topology is also useful for showing that equilibria exist in the
case of an infinite number of agents, which requires a non-standard
argument (Theorem~\ref{thm:equi-exists}). This technique should be
applicable to the analysis of a large range of repeated, discounted
games on networks.

\section{Non-learning}
\label{sec:non-learning}
\begin{figure}[h]
  \centering
  \includegraphics[scale=0.8]{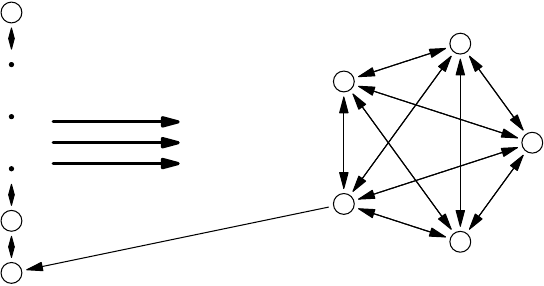}
  \caption{\label{fig:royal-family} The Royal Family. Each member of
    the public (on the left), observes each royal (on the right), as
    well as her next door neighbors. The royals observe each other,
    and one royal observes one member of the public. }
\end{figure}

We provide two example of non-egalitarian graphs in which the agents
do not learn. In the first example (Figure~\ref{fig:royal-family}),
inspired by Bala and Goyal's royal family graph~\cite{BalaGoyal:96},
the social network has two groups of agents: a ``royal family'' clique
of $R$ agents who all observe each other, and $n$ agents - the
``public'' - who are connected in an undirected chain, and
additionally can observe all the agents in the royal family. Finally,
a single member of the royal family observes one of the public, so
that the graph is connected\footnote{The graph is, in fact, {\em
    strongly connected}, meaning that there is a directed path
  connecting every ordered pair of agents.}. We think of $R$ as fixed
and consider the case of arbitrarily large $n$, or even infinite $n$.

While this graph satisfies condition (1) of egalitarianism, it
violates condition (2). Therefore,
Theorem~\ref{thm:learning-finite-intro} does not apply.  Indeed, in
the online appendix we construct an equilibrium for the game on this
network, in which the agents of the public ignore their own private
signals after observing the first action of the royal family, which
provides a much stronger indication of the correct action.  However,
the probability that the royal family is wrong is independent of $n$:
since the size of the royal family is fixed, with some fixed
probability every one of its members is mislead by her private signal
to choose the wrong action in the first period. Hence, regardless of
how large society is, there is a fixed probability that learning does
not occur.

\begin{figure}[h]
  \centering
  \includegraphics[scale=0.8]{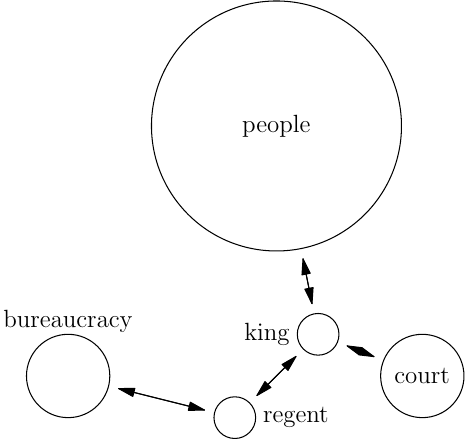}
  \caption{\label{fig:mad-royal-family-intro} The mad king. In this
    social network all edges are bi-directional. Each member of the
    people is connected only to the king, as is each member of the
    court. The members of the bureaucracy are connected only to the
    regent.}
\end{figure}

We construct a second example of non-learning which we call ``the mad
king''. Here, the graph is undirected, so that whenever $i$ observes
$j$ then $j$ observes $i$; the graph therefore satisfies condition (2)
of egalitarianism with $L=1$, but not condition (1).  The graph (see
Figure~\ref{fig:mad-royal-family-intro}) consists of five types: the
people, the king, the regent, the court and the bureaucracy. There is
one king, one regent, a fixed number of members of the court and a
fixed number of members of the bureaucracy, which is much larger than
the court. The number of people is arbitrarily large. They are
connected as follows:
\begin{itemize}
\item The king  is connected to the regent, the court and the people.
\item The regent is connected to the king and to the bureaucracy.
\item The members of the court are each connected only to the king.
\item The members of the bureaucracy are each connected only to the
  regent.
\item Importantly, the people are each connected only to the king, and
  not to each other.
\end{itemize}

For an appropriate choice of private signal distributions and discount
factor, we construct an equilibrium in which all agents act myopically
in the first two rounds, except for the people, who choose the
constant action $0$. This is enforced by a threat from the king, who,
if any of the people deviate, will always play $1$, denying them any
information he has learned; the prize for complying is the exposure to
a well informed action which first aggregates the information
available to the court, and later aggregates the information available
to the (larger) bureaucracy.  The result is that the information in
the people's private signals is lost, and so we have non-learning with
probability bounded away from $0$, for graphs of arbitrarily large
size.

The equilibrium path can be succinctly described as follows; we
provide a complete description in the online appendix:
\begin{itemize}
\item The members of the bureaucracy act myopically in round 0, as do
  the members of the court.
\item The regent, who learns by observing the bureaucracy, acts
  myopically at time 0 and at time 1. Therefore, and since the
  bureaucracy is large, his action at time 1 will be correct with some
  fixed high probability.
\item The king acts myopically in round 0. At round 1, after having
  learned from the court's actions, the king again acts myopically,
  unless any of the people chose action 1 at round 0, in which case he
  chooses action 1 at this time and henceforth.
\item The people choose action 0 in round 0. They have no incentive to
  deviate, since they stand to learn much from the king's actions,
  which, at the next round, will aggregate the information in the
  court's actions.
\item By round 2, the king has learned from the regent's well informed
  action of round 1. He therefore, at round 2, emulates the regent's
  action of round 1, unless any of the people chose action 1 at rounds
  0 or 1, in which case he again chooses action 1 at this time and
  henceforth.
\item In round 1 the people again choose action 0. They again have no
  reason to deviate, this time because they wish to learn the regent's
  action, through the king; this information - which originates from
  the bureaucracy - is much more precise than that which the the king
  collected from the court in the previous round and reveals to them
  in this round.
\item At round 2 (and henceforth) the people emulate the king's
  previous action, and therefore the king will not learn from them.
\end{itemize}
It follows that the private signals of the people are lost, and so,
regardless of the number of people, there is a fixed probability of
non-learning.

We were not able to prove - or to disprove - that this equilibrium is
a perfect Bayesian equilibrium. However, it intuitively seems likely
that if the people were to deviate from the equilibrium, then the king
would not have an incentive to carry out his threat. If this intuition
holds then this is not a perfect Bayesian equilibrium.

An interesting phenomenon is that on this graph, the agents do learn
$S$ with high probability when they discount the future sufficiently,
or in the limiting case that they are myopic (i.e., fully discount the
future).  This is thus an example - and perhaps a counter-intuitive
one - of how strategic agents may learn less effectively than
myopic ones.

\section{Model}
\label{sec:main-defs}

\subsection{Informational structure}
The structure of the private information available to the agents is
the standard one used in the herding literature (see, e.g., Smith and
S{\o}rensen~\cite{smith2000pathological}).

We denote by $V$ the set of agents, which we take to equal
$\{1,2,\ldots,n\}$ in the finite case and $\N=\{1,2,\ldots\}$ in the
(countably) infinite case.  Let $\{0,1\}$ be the set of possible
values of the {\em state of nature} $S$, and let
$\P{S=1}=\P{S=0}=1/2$. Let $\Omega$ be a measurable space, called the
space of {\em private signals}.  Let $\psignal_i \in \Omega$ be agent
$i$'s private signal, and denote $\bar{\psignal} =
(\psignal_1,\psignal_2,\ldots)$.  Fix $\mu_0$ and $\mu_1$, two
mutually absolutely continuous probability measures on
$\Omega$. Conditioned on $S=0$, let $\psignal_i$ be i.i.d.\ $\mu_0$,
and conditioned on $S=1$ let $\psignal_i$ be i.i.d.\ $\mu_1$.

The assumption that $\P{S=1}=1/2$ can be relaxed; in particular, for
every choice of private signals there exist $p_1 < 1/2 < p_2$ such
that our results apply when $\P{S=1}$ is taken to be in $(p_1,p_2)$.
However, when agents are myopic (or more generally discount the future
enough), when priors are skewed, and when signals are weak, then
regardless of the graph, and in any equilibrium, agents will disregard
their private signals and only play the more a priori probable
action. Indeed, if for example the prior is $\P{S=1}=0.8$, then for
weak enough private signals it will be the case that
$\CondP{S=1}{\psignal_i} > 0.7$ with probability one. Myopic agents
always choose the action that they deem more likely to equal the state
of the world, and therefore will choose $1$, as will agents who are
not myopic but sufficiently discount the future. It follows that in
this case observing others' actions will reveal no information, and no
learning will occur. Therefore, to ensure learning, one must impose
some conditions on the prior and the strength of the private signals;
for example, a sufficient condition would be that
$\CondP{S=1}{\psignal_i}$ has positive probability of both being above
half and of being below half. To avoid encumbering this paper with an
additional layer of technical complexity, we focus on the case in
which $\P{S=1}=\P{S=0}=1/2$.

Agent $i$'s {\em private belief} $\belief_i$ is the probability that
the state of the world is $1$, given $i$'s private signal:
\begin{align*}
  \belief_i = \CondP{S=1}{\psignal_i}.
\end{align*}
Since $\belief_i$ is a sufficient statistic for $S$ given
$\psignal_i$, we assume below without loss of generality that an
agent's actions depend on $\psignal_i$ only through $\belief_i$ (see,
e.g., Smith and S{\o}rensen~\cite{smith2000pathological}).

We consider only $\mu_0$ and $\mu_1$ such that the distribution of
$\belief_i$ is non-atomic. This is the condition that we refer to
above as {\em non-atomic private beliefs}. This is an additional
restriction that we impose, beyond what is standard in the herding
literature.

\subsection{The social network}
The agents' social network defines which of them observe the actions
of which others. We do not assume that this is a symmetric relation:
it may be that $i$ observes $j$ while $j$ doesn't observe
$i$. Formally, the social network $G=(V,E)$ is a directed graph: $V$
is the set of agents, and $E$ is a relation on $V$, or a subset of 
the set of ordered pairs $V \times V$.  The set of neighbors of $i \in
V$ is
\begin{align*}
  \neigh{i} = \{j :\: (i,j) \in E\},
\end{align*}
and we consider only graphs in which $i \in N(i)$; that is, we require
that an agent observes her own actions. The {\em out-degree} of $i$ is
given by $|\neigh{i}|$, and will always be finite. This means that an
agent observes the actions of a finite number of other agents. We do
allow infinite {\em in-degrees}; this corresponds to agents whose
actions are observed by infinitely many other agents.

Let $G=(V,E)$ be a directed graph.  A (directed) {\em path} of length
$k$ from $i \in V$ to $j \in V$ in $G$ is sequence of $k+1$ nodes
$i_1,\ldots,i_{k+1}$ such that $(i_n,i_{n+1}) \in E$ for
$n=1,\ldots,k$, and where $i_1=i$ and $i_{k+1}=j$.

A directed graph is {\em strongly connected} if there exists a
directed path between every ordered pair of nodes; we restrict our
attention to such graphs.  Strong connectedness is natural in the
contexts of agreement and learning; as an extreme example, consider a
graph in which some agent observes no-one. In this graph we cannot
hope for that agent to learn the state of nature.

A directed graph $G$ is {\em $L$-locally-connected} if, for each
$(i,j) \in E$, there exists a path of length at most $L$ in $G$ from
$j$ to $i$.  Equivalently, $G$ is $L$-locally-connected if whenever
there exists a path of length $k$ from $i$ to $j$, there exists a path
of length at most $L \cdot k$ from $j$ back to $i$.  Note that
$1$-locally-connected graphs are commonly known as undirected graphs.

As defined above, a graph is said to be $(d,L)$-egalitarian if all
out-degrees are bounded by $d$, and if it is $L$-locally-connected.

\subsection{The game}
To model the agents' strategic behavior we consider the following game
of incomplete information. This framework, with some variations, has
been previously used, for example, by Gale and
Kariv~\cite{GaleKariv:03} and Rosenberg, Solan and
Vieille~\cite{rosenberg2009informational}.

We consider the discrete time periods $t=0,1,2,\ldots$, where in each
period each agent $i \in V$ has to choose one of the actions in
$\{0,1\}$. The information available to $i$ at time $t$ is her own
private signal (of which the relevant information is her private
belief, taking values in $[0,1]$), and the actions of her neighbors in
previous time periods, taking values in $\{0,1\}^{|\neigh{i}|\cdot
  t}$. This action is hence calculated by some function from
$[0,1]\times\{0,1\}^{|\neigh{i}|\cdot t}$ to $\{0,1\}$.

A {\em pure strategy at time $t$} of an agent $i \in V$ is therefore a
Borel-measurable function $q^i_t :
[0,1]\times\{0,1\}^{|\neigh{i}|\cdot t} \to \{0,1\}$. A {\em pure
  strategy} of an agent $i$ is the sequence of functions
$q^i=(q^i_0,q^i_1,\ldots)$, where $q^i_t$ is $i$'s pure strategy at
time $t$. We endow the space of pure strategies with the topology
derived from the weak topology on functions from $[0,1]$ to $\{0,1\}$.

A {\em mixed strategy} $\strat^i$ of agent $i$ is a
pure-strategy-valued random variable; this is the standard notion of a
mixed strategy, and we shall henceforth refer to mixed strategies
simply as strategies. A (mixed) {\em strategy profile} is a set of
strategies $\stratp=\{\strat^i:i \in V\}$, where the random variables
$\strat^i$ are independent of each other and of the private
signals. 

The {\em action} of agent $i$ at time $t$ is denoted by $\action^i_t
\in \{0,1\}$. Denote the {\em history} of actions of the neighbors of
$i$ before time $t$ by $\hist{i}{t} = \{\action^j_s:s < t, j \in
\neigh{i}\}$; this depends on the social network $G$.  The action that
agent $i$ plays at time $t$ under strategy profile $\stratp$ is
\begin{align*}
  \action^i_t = \action^i_t(G,\stratp) = \strat^i_t\left(\belief_i,
  \action^{\neigh{i}}_{[0,t)}\right).
\end{align*}

Note again that we (without loss of generality) limit the action to be
a function of the private belief $\belief_i$, as opposed to the
private signal $\psignal_i$.

Let $0 < \disc < 1$ denote the agents' common {\em discount factor}.
Given a social network $G$ and strategy profile $\stratp$, agent $i$'s
{\em stage utility at time $t$}, $U_{i,t}$, is $1$ if her action
matches $S$, and $0$ otherwise:
\begin{align*}
  U_{i,t} = U_{i,t}(G,\stratp) = \ind{\action^i_t(G,\stratp) = S}.
\end{align*}
Her {\em expected stage utility at time $t$}, $\util_{i,t}$, is
therefore given by
\begin{align*}
  \util_{i,t} = \util_{i,t}(G,\stratp) = \E{U_{i,t}(G,\stratp)} =
  \P{\action^i_t(G,\stratp) = S}.
\end{align*}
Agent $i$'s {\em expected utility} $\util_i$ is given by
\begin{align*}
  \util_i = \util_i(G, \stratp) = (1-\disc)\sum_{t =
    0}^\infty\disc^t\util_{i,t}(G,\stratp).
\end{align*}
  
Note that $\util_i \in [0,1]$, due to the normalization factor
$(1-\disc)$.  A {\em game} $\cG$ is a 4-tuple $(\mu_0,\mu_1,\disc,G)$
consisting of two measures, a discount factor and a social network
graph, satisfying the conditions of the definitions above.

\subsection{Equilibria}

Our equilibrium concept is the standard Nash equilibrium in games of
incomplete information: $\stratp$ is an equilibrium if no agent can
improve her expected utility $\util_i(\stratp)$ by deviating from
$\stratp$.

Formally, in a game $\cG = (\mu_0,\mu_1,\disc,G)$, strategy profile
$\stratp$ is an equilibrium if, for every agent $i \in V$ it
holds that
\begin{align*}
  \util_i(G,\stratp) \geq \util_i(G,\stratpx),
\end{align*}
for any $\stratpx$ such that $\stratx^j=\strat^j$ for all $j \neq i$
in $V$.

\section{Agreement}
\label{sec:agreement}
Let the {\em infinite action set} $\optset_i$ of agent $i$ be defined
by
\begin{align*}
  \optset_i = \optset_i(G,\stratp) = \{s \in
  \{0,1\}:\action^i_t(G,\stratp)=s \mbox{ for infinitely many values of
    $t$}\}.
\end{align*}
There could be more than one action that $i$ takes infinitely
often. In that case we write $\optset_i = \{0,1\}$. Otherwise, with a
slight abuse of notation, we write $\optset_i=0$ or $\optset_i=1$, as
appropriate.

In this section we show that the agents reach consensus in any graph,
in the following sense:
\begin{theorem}
  \label{thm:agreement}
  Let $\cG$ be a game with either finitely many players or countably
  infinitely many players, and let $\stratp$ be an equilibrium
  strategy profile of $\cG$. Then, with probability one,
  $\optset_i=\optset_j$ for all agents $i,j \in V$.
\end{theorem}
This theorem is a crucial ingredient in the proof of the main result
of this article. Indeed, learning occurs if $\optset_i=S$ for all $i$,
and so a prerequisite is that $\optset_i=\optset_j$ for all $i,j$.

Recall that a strategy of agent $i$ at time $t$ is a function of her
private belief $\belief_i$ and the actions of her neighbors in
previous time periods, $\hist{i}{t}$. Hence we can think of the
sigma-algebra generated by these random variables as the ``information
available to agent $i$ at time $t$''. Denote the information available
to agent $i$ at time $t$ by
\begin{align*}
  \info^i_t = \info^i_t(G,\stratp) =
  \sigma\left(\belief_i,\strat^i,\hist{i}{t}\right),
\end{align*}
and denote by
\begin{align*}
  \info^i_{\infty} = \info^i_{\infty}(G,\stratp) =
  \sigma\left(\cup_{t=0}^\infty\info^i_t\right)
\end{align*}
the information available to agent $i$ at the limit $t \to \infty$.
Note that $\info^i_t$ includes the sigma-algebra generated by $i$'s
private belief, the actions of $i$'s neighbors before time $t$, and
$i$'s pure strategy; $i$ knows which pure strategy she has chosen.

Since the expected stage utility of action $s$ at time $t$ is
$\P{s=S}$, a myopic agent would take an action $s$ in $\{0,1\}$ that
maximizes $\CondP{s=S}{\info^i_t}$. This motivates the following
definition. Denote the best response of agent $i$ at time $t$ by
\begin{align*}
  \bestres^i_t = \bestres^i_t(G,\stratp) = \argmax_{s \in
    \{0,1\}}\CondP{s=S}{\info^i_t(G,\stratp)}.
\end{align*}
Likewise denote the {\em set} of best responses of agent $i$ at the
limit $t \to \infty$ by
\begin{align*}
  \bestres^i_{\infty} = \bestres^i_{\infty}(G,\stratp) = \argmax_{s
    \in \{0,1\}}\CondP{s=S}{\info^i_{\infty}}.
\end{align*}

At any time $t$ there is indeed almost surely only one action that
maximizes $\CondP{s=S}{\info^i_t(G,\stratp)}$, since we require that
the distribution of private beliefs be non atomic. This does not
necessarily hold at the limit $t \to \infty$, and so we let
$\bestres^i_{\infty}$ take the values $0$, $1$ or $\{0,1\}$. Note that
a reasonable conjecture is that the probability that
$\bestres^i_{\infty} = \{0,1\}$ is zero, but we are not able to prove
this. This does not, however, prevent us from proving our results, but
it does complicate the proofs.

The following theorem is a restatement, in our notation, of
Proposition 2.1 in Rosenberg, Solan and
Vieille~\cite{rosenberg2009informational}.
\begin{theorem}[Rosenberg, Solan and Vieille]
  \label{thm:rsv1}
  For any agent $i$ it holds that $\optset_i \subseteq
  \bestres^i_{\infty}$ almost surely, in any equilibrium.
\end{theorem}
That is, any action that $i$ takes infinitely often is optimal, given
all the information agent $i$ eventually learns. Note that this
theorem is stated in~\cite{rosenberg2009informational} for a finite
number of agents. However, a careful reading of the proof reveals that
it holds equally for a countably infinite set of agents. The same
holds for their Theorem 2.3, in which they further prove the following
agreement result.
\begin{theorem}[Rosenberg, Solan and Vieille]
  \label{thm:rsv-agreement}
  Let $j$ be a neighbor of $i$. Then $\optset_j \subseteq
  \bestres^i_{\infty}$ almost surely, in any equilibrium.
\end{theorem}
Equivalently, if $i$ observes $j$, and $j$ takes an action $a$
infinitely often, then $a$ is an optimal action for $i$. If we could
show that $\bestres^i_{\infty} = \optset_i$ for all $i$, it would
follow from these two theorems, and from the fact that the graph is
strongly connected, that $\optset_i = \optset_j$ for all agents $i$
and $j$; the agents would agree on their optimal action sets. This is
precisely what we show in Theorem~\ref{thm:bestres-optset}. Our
agreement theorem (Theorem~\ref{thm:agreement}) is a direct consequence.

\section{Topologies on graphs and strategy profiles}
\label{sec:topologies}
\subsection{Rooted graphs and their topology}
A {\em rooted graph} is a pair $(G,i)$, where $G=(V,E)$ is a directed
graph, and $i \in V$ is a vertex in $G$.

Rooted graphs are a basic mathematical concept, and are important to
the understanding of this game. This section starts with some basic
definitions, continues with the definition of a metric topology on
rooted graphs, and culminates in a novel theorem on compactness in
this topology, which may be of independent interest. In this we follow
our previous work~\cite{mossel2012asymptotic}, which builds on the
work of others such as Benjamini and
Schramm~\cite{benjamini2011recurrence} and Aldous and
Steele~\cite{aldous2003objective}.

Intuitively, a rooted graph is a graph, as seen from the ``point of
view'' of a particular vertex - the root. Two rooted graphs will be
close in our topology if the two graphs are similar, as seen from the
roots.

Before defining our topology we will need a number of standard
definitions.  Let $G=(V,E)$ and $G'=(V',E')$ be graphs, and
let $(G,i)$ and $(G',i')$ be rooted graphs. A {\em rooted graph
  isomorphism} between $(G,i)$ and $(G',i')$ is a bijection $h:V \to
V'$ such that
\begin{enumerate}
\item $h(i) = i'$.
\item $(j,k) \in E \,\Leftrightarrow\, (h(j),h(k)) \in E'$.
\end{enumerate}
If there exists a rooted graph isomorphism between $(G,i)$ and
$(G',i)$ then we say that they are isomorphic, and write $(G,i) \cong
(G',i')$.  Informally, isomorphic graphs cannot be told apart when
vertex labels are removed; equivalently, one can be turned into the
other by an appropriate renaming of the vertices. The isomorphism
class of $(G,i)$ is the set of rooted graphs that are isomorphic to
it, and will be denoted by $[G,i]$.

Let $j,k$ be vertices in a graph $G$. Denote by $\dist(j,k)$
the length of the shortest (directed) path from $j$ to $k$.  In general,
$\dist(j,k) \neq \dist(k,j)$, since the graph is directed.  The
(directed) {\em ball} $B_r(G,i)$ of radius $r$ of the rooted graph
$(G,i)$ is the rooted graph, with root $i$, induced in $G$ by the set
of vertices $\{j \in V\,:\,\dist(i,j) \leq r\}$.

We now proceed to define our topology on the space of isomorphism
classes of strongly connected rooted graphs, which is an extension of
the Benjamini-Schramm~\cite{benjamini2011recurrence} topology on
undirected graphs. We define this topology by a metric\footnote{ This
  definition applies, in fact, to a larger class of directed graphs: a
  rooted graph $(G,i)$ is {\em weakly connected} if there is a
  directed path from $i$ to each other vertex in the graph. Note that
  indeed a strongly connected graph is necessarily weakly connected,
  but not vice versa. Note also that a rooted graph $(G,i)$ is weakly
  connected if and only if for every vertex $j$ there exists an $r$
  such that $j$ is in $B_r(G,i)$.  }.

Let $[G',i']$ and $[G,i]$ be isomorphism classes of strongly connected
rooted graphs. The distance $D([G,i],[G',i'])$ is defined by
\begin{align}
  \label{eq:graph-metric-def}
  D([G,i],[G',i']) = \inf \{2^{-r}\,:\,B_r(G,i) \cong B_r(G',i')\}.
\end{align}
That is, the larger the radius around the roots in which the graphs are
isomorphic, the closer they are. In fact, the quantitative dependence
of $D(\cdot,\cdot)$ on $r$ (exponential in our definition) will not be
of importance here, as we shall only be interested in the topology
induced by this metric.

It is straightforward to show that $D(\cdot,\cdot)$ is well defined; a
standard diagonalization argument (which we use repeatedly in this
article) is needed to show that it is indeed a metric rather than a
pseudometric (Claim~\ref{claim:metric}). The assumption of strong
connectivity is crucial here, since $D(\cdot,\cdot)$ is otherwise a
pseudometric.

Let $\scg$ be the set of isomorphism classes of {\em strongly
  connected} rooted graphs. This set is a topological space when
equipped with the topology induced by the metric $D(\cdot,\cdot)$.
Given a strongly connected graph $G$, let $\cR(G) \subset \scg$ be the
set of all rooted graph isomorphism classes of the form $[G,i]$, for
$i$ a vertex in $G$. This can be thought of as the set of ``points of
view'' in the graph $G$. The notion of $(d,L)$-egalitarianism now
arises naturally, in the sense that the number of ``approximate points
of view'' in $G$ is finite if and only if $G$ is egalitarian. This is
formalized in the following lemma.
\begin{lemma}
  \label{cor:compact}
  Let $G$ be a strongly connected graph. Then the closure of $\cR(G)$
  is compact in $\scg$ if and only if $G$ is $(d,L)$-egalitarian, for
  some $d$ and $L$.
\end{lemma}

We would like to suggest that Lemma~\ref{cor:compact}, which we prove
in Appendix~\ref{app:rooted-graphs-proofs}, may be of independent
mathematical interest, as it extends the well understood notion of
compactness in undirected graphs to directed, strongly connected
graphs.

\subsection{The space of rooted graph strategy profiles and its
  topology}
\label{sec:rooted-graph-strategy-profiles}
In this section we use the above topology on rooted graphs to
construct a topology on what we call {\em rooted graph strategy
  profiles}. This will be the main tool at our disposal in proving
both the existence of equilibria, and our main result,
Theorem~\ref{thm:learning-finite-intro}. Intuitively, a {\em rooted
  graph strategy profile} will be a graph, together with a strategy
profile, as seen from the point of view of the root. As in the case of
rooted graphs, two points in this space will be close if they look
alike from the points of view of the roots.

Let $G=(V,E)$ and $G'=(V',E')$ be strongly connected directed graphs,
and let $(G,i), (G',i') \in \scg$ be rooted graphs. Let $\stratp$ and
$\stratpx$ be strategy profiles for the agents in $V$ and $V'$,
respectively. We say that the triplet $(G,i,\stratp)$ is equivalent to
the triplet $(G',i',\stratpx)$ if there exists a rooted graph
isomorphism $h$ from $(G,i)$ to $(G',i')$ such that $\stratp^j =
\stratpx^{h(j)}$ for all $j \in V$. The {\em rooted graph strategy
  profiles} $\gs$ are the set of equivalence classes induced by this
equivalence relation. We denote an element of $\gs$ by
$[G,i,\stratp]$.

In Appendix~\ref{app:strategy-topology} we apply the classical work of
Milgrom and Weber~\cite{milgrom1985distributional} to define a metric
$d$ on a single agent's strategy space, with the property that when
the number of agents is finite then utilities are continuous in the
induced topology.

We use this metric, and the metric of rooted graphs to define a metric
on rooted graph strategy profiles. Intuitively, $[G,i,\stratp]$ and
$[G',i',\stratpx]$ will be close in this metric if, in a large radius
around $i$ and $i'$, it holds both that the graphs are isomorphic and
that the strategies are similar.

Let $d$ be a metric on a single agent's strategy space. Let $i$ and
$i'$ be agents in graphs $G$ and $G'$, respectively. We can use $d$ as
a metric between their strategies, as long as we uniquely identify
each neighbor of one with a neighbor of the other. Let $h$ be a
bijection between $\neigh{i'}$ and $\neigh{i}$. Then
$d_h(\strat^i,\strat^{i'})$ will denote the distance thus defined
between $\strat^i$ and $\strat^{i'}$.

We next define $D_r(\cdot,\cdot)$, a pseudometric on graph strategy
profiles which only takes into account the graph and the strategies at
balls of radius $r$ around the root. Two graph strategy profiles are
close in $D_r$ if (1) these balls are isomorphic, so that agents in
these balls can be identified, and if (2) under some such
identification, identified agents have similar strategies. This is a
pseudometric rather than a metric since there could be two graph
strategy profiles that are at distance $0$ under $D_r$, but are not
identical; differences will, however, occur only at distances that are
larger than $r$ from the roots.

Let $[G,i,\stratp]$ and $[G',i',\stratpx]$ be rooted graph
strategies. For $r \in \N$, let $H(r)$ be the (perhaps empty) set of
rooted graph isomorphisms between $B_r(G,i)$ and $B_r(G',i')$. Let
\begin{align*}
  D_r\Big([G,i,\stratp],[G',i',\stratpx]\Big) = \min_{h \in
    H(r+1)}\max_{j \in B_r(G,i)}d_h(\strat^j,\stratx^{h(j)}),
\end{align*}
when $H(r+1)$ is non-empty, and $1$ otherwise. The choice of $h \in
H(r+1)$ and then $j \in B_r(G,i)$ guarantees that $h$ is a bijection
from the set of neighbors of $j$ to the set of neighbors of $h(j)$.

Finally, define the metric $D([G,i,\stratp],[G',i',\stratpx])$ by
\begin{align}
  \label{eq:D-def}
  D\Big([G,i,\stratp],[G',i',\stratpx]\Big) = \inf_{r \in
    \N}\left\{\max\left\{2^{-r},D_r([G,i,\stratp],[G',i',\stratpx])\right\}\right\}.
\end{align}
Note that $D$ will be small whenever $D_r$ is small for large $r$. It
is straightforward (if tedious) to show that $D(\cdot,\cdot)$ is
indeed a well defined metric.

\subsection{Properties of the space of rooted graph strategy profiles}
Two rooted graph strategy profiles will be close in the topology
induced by $D$ if, in a large neighborhood of the roots, it holds both
that the graphs are isomorphic, and also that the strategies are
similar. This captures the root's ``point of view'' {\em of the entire
  strategy profile}.

While many possible topologies may have this property, this topology
has some technical features that make it a useful analytical tool.
First, expected utilities are continuous is this topology.  Formally,
let the {\em utility map} $\util : \gs \to \R$ be given by
\begin{align*}
  \util([G,i,\stratp]) = \util_i(G,\stratp).
\end{align*}
This is a straightforward recasting of the previous definition of
expected utility into the language of rooted graph strategy spaces. In
Lemma~\ref{thm:utility-cont} we show that $u : \gs \to \R$ is
continuous; this follows from the fact that payoff is discounted, and
so the strategies of far away agents have only a small effect on an
agent's utility. Another property of this topology that makes it
applicable is that the set of {\em equilibrium} rooted graph strategy
profiles is closed (Lemma~\ref{thm:equi-conv-equi}). These properties
are also instrumental in proving that equilibria exist
(Theorem~\ref{thm:equi-exists}).

Additionally, the probability of learning is lower semi-continuous in
this topology.  Let the {\em probability of learning} map $\learnprob
: \gs \to \R$ be given by
\begin{align*}
  \learnprob([G,i,\stratp]) = \lim_{t \to
    \infty}\P{\action^i_t(G,\stratp)=S}.
\end{align*}
In Section~\ref{app:learnprob} we prove that $\learnprob$ is well
defined and that it is lower semi-continuous
(Theorem~\ref{thm:p-semi-cont}).  We also show that
$\learnprob([G,i,\stratp])=1$ if and only if the agents learn; i.e.,
if and only if $\lim_t\action^j_t=S$ almost surely for all agents $j$
in $G$ (Claim~\ref{clm:p-learning}).

Finally, if $G$ is an egalitarian graph, then the set of rooted graph
strategy profiles on $G$ is precompact
(Claim~\ref{thm:conv-conv}). Intuitively, this means that when $G$ is
egalitarian then not only are there finitely many approximate points
of view of the graph (as discussed above), but also just finitely many
approximate points of view of the strategy profile.

\section{Learning}
\label{sec:learning}
\subsection{Learning on infinite egalitarian graphs}

Let $G$ be an infinite, connected, $(d,L)$-egalitarian graph, and let
$\stratp$ be an equilibrium strategy profile. In this section we show
that all agents learns $S$ almost surely.

Recall that all agents converge to the same (random) action or set of
actions. Denote by $\mapS_{\infty}$ the random variable that is equal
to $0$ if all agents converge to $0$ and is equal to $1$ if all agents
converge to $1$, or if they all do not converge. Our choice of
notation here follows from the fact that $\mapS_{\infty}$ is a maximum
a posteriori (MAP) estimator of any particular agent, given all that
it learns: namely, the probability that an agent learns $S$ is equal
to the probability that $\mapS_{\infty}$ equals $S$
(Claim~\ref{clm:learnprob-bestres}). Since the private signals are
informative, $\mapS_{\infty}=S$ with probability which is strictly
greater than one half (Claim~\ref{clm:p-half}), so $\mapS_\infty$ is a
non-trivial estimator of $S$.

Note that $\mapS_{\infty}$ is measurable in the sequence of every
agent's actions. Hence each agent eventually learns it, or something
``close to it'' at large finite times: formally, for every $\delta >
0$ there will be a time $t$ and random variable
$\mapS_{\infty}^{i,\delta}$ that can be calculated by $i$ at time $t$,
and such that $\P{\mapS_{\infty}^{i,\delta} = \mapS_{\infty}} > 1-\delta$.

Now, $\mapS_{\infty}$ is a deterministic function of the agents'
private signals and pure strategies. Hence (e.g., by the martingale
convergence theorem) $\mapS_{\infty}$ is an {\em almost deterministic}
function of the private signals and pure strategies of a large but
{\em finite} group of agents. Formally, for every $\eps > 0$ there is
a random variable $\mapS_{\infty}^\eps$ that depends only on the
private signals and pure strategies of some finite set of agents
$V^\eps$, and such that $\P{\mapS_{\infty}^\eps =
  \mapS_{\infty}}>1-\eps$.

Let $i$ be an agent who is far away (in graph distance) from $V^\eps$,
so that the nearest member of $V^\eps$ is at {\em distance} at least
$t$ from $i$. Then everything that $i$ observes up to {\em time} $t$
is independent of $\mapS_{\infty}^\eps$, and hence ``approximately
independent'' of $\mapS_{\infty}$ (Claim~\ref{clm:maps-local}); we
formalize a notion of ``approximate independence'' in
Section~\ref{app:delta-ind}.

Now, as we note above, $i$ eventually learns $\mapS_{\infty}$ (or more
precisely an estimator $\mapS_{\infty}^{i,\delta}$ that is equal to
$\mapS_{\infty}$ with high probability), gaining a new estimator of
$S$ which is (approximately) independent of any estimators that it has
learned up to time $t$.  What we have so far outlined can thus be
summarized informally as follows: the estimator $\mapS_\infty$ is
``decided upon'' by a finite group of agents. When those far away
eventually learn it they gain a new, approximately independent
estimator of $S$.

We apply this argument inductively, relying crucially on the fact that
the space of rooted graph strategies on $G$ is precompact: Assume by
induction that for every ``point of view'' $[H,j,\stratpx]$ in the
closure of this space there is an agent $i$ in $H$ with $k-1$
approximately independent estimators of $S$ by some time $t$. By
compactness and the infinitude of $G$, there are infinitely many
agents in $G$ whose points of view are approximately equal to that of
such an agent $i$. These will all also have $k-1$ approximately
independent estimators of $S$ by time $t$.  Some (in fact, almost all)
of these agents will be sufficiently far from $V^\eps$. These will
then gain a new estimator when they eventually learn $\mapS_\infty$.

Hence in egalitarian graph, for any $k$ and any degree of
approximation, there will always be an agent who, given enough time,
will accumulate $k$ approximately independent estimators of $S$
(Lemma~\ref{lem:independent-ests}). A standard concentration of
measure inequality then guarantees that the agent's probability of
learning will be approximately $1$
(Theorem~\ref{thm:compact-learning}). This proves that the agents
learn on infinite graphs.

\subsection{Learning on finite egalitarian graphs and the proof of Theorem~\ref{thm:learning-finite-intro}}

We reduce the case of finite graphs to that of infinite graphs, thus
proving our main theorem.
\begin{maintheorem}
  \label{thm:learning-finite-intro}
  Fix the distributions of the agents' private signals, with
  non-atomic private beliefs. Fix also a discount factor $\disc \in
  (0,1)$, and positive integers $L$ and $d$.  Then in any connected,
  $(d,L)$-egalitarian, countably infinite network
  \begin{align*}
    \P{\mbox{all agents learn } S } =1
  \end{align*}
  in any equilibrium. Furthermore, for every $\eps>0$ there exists an
  $n$ such that for any connected, $(d,L)$-egalitarian network with at
  least $n$ agents
  \begin{align*}
    \P{\mbox{all agents learn } S } \geq 1-\eps,
  \end{align*}
  in any equilibrium.
\end{maintheorem}

Given a set of graphs $\cK$, let $\cR(\cK)$ be the set of rooted
graphs $[G,i]$ such that $G \in \cK$. Let $\eqs(\cK)$ be the set of
equilibrium strategy profiles $[G,i,\stratp]$ such that $G \in \cK$.
\begin{proof}[Proof of Theorem~\ref{thm:learning-finite-intro}]
  Let $G$ be a $(d,L)$-egalitarian graph. The case that $G$ is
  infinite is treated in Theorem~\ref{thm:compact-learning}.

  We hence consider finite graphs.  Let $\cK_n$ be the set of
  $L$-locally-connected, degree $d$ graphs with $n$ vertices. Since
  $\cK_n$ is finite then $\cR(\cK_n)$ is finite and hence compact. It
  follows that $\eqs(\cK_n)$ is also compact
  (Claim~\ref{clm:cB-compact}). Since the map $\learnprob$ is lower
  semi-continuous it attains a minimum on $\eqs(\cK_n)$. Let
  $[G_n,i_n,\stratp_n]$ be a minimum point, and denote $q(n) =
  \learnprob([G_n,i_n,\stratp_n])$. We will prove the claim by showing
  that $\lim_n q(n)=1$. Let $\{q(n_k)\}_{k=1}^{\infty}$ be a
  subsequence such that $\lim_k q(n_k) = \liminf_n q(n)$.
  
  Since the set of $(d,L)$-egalitarian graphs is compact
  (Theorem~\ref{thm:cBLd-compact}), by again invoking
  Claim~\ref{clm:cB-compact}, we have that the sequence
  $\{[G_{n_k},i_{n_k},\stratp_{n_k}]\}_{k=1}^{\infty}$ has a
  converging subsequence that must converge to some {\em infinite}
  $L$-locally-connected, degree $d$ equilibrium graph strategy
  $[G,i,\stratp]$. By the above, we have that
  $\learnprob([G,i,\stratp]) = 1$, and so, by the lower
  semi-continuity of $p$, it follows that
  \begin{align*}
    \liminf_{n \to \infty}q(n) = \lim_k q(n_k) = \lim_{k \to
      \infty}\learnprob([G_{n_k},i_{n_k},\stratp_{n_k}]) \geq
    \learnprob([G,i,\stratp])=1.
  \end{align*}
\end{proof}

\section{Conclusion}
\label{sec:conclusion}

\subsection{Summary}
Learning on social networks by observing the actions of others is a
natural phenomenon that has been studied extensively in the
literature. However, the question of how strategic agents fare has
been largely ignored. We tackle this problem in a standard framework
of a discounted game of incomplete information and conditionally
independent private signals.

We show that on some networks agents learn in every equilibrium, and
that they do not necessarily learn on others. The geometric condition
of learning is one of egalitarianism, and is similar in spirit to
conditions of learning identified in some boundedly-rational models.

\subsection{Extensions and open problems}
Our techniques, by their topological nature, give only asymptotic
results: we show that the probability that agents learn on a
$(d,L)$-egalitarian graph with $n$ agents tends to one. It may be
interesting to study the rate at which this happens, but our
techniques do not seem to apply to this question.

Natural extensions of our model include those in which agents do not
act synchronously, and those in which the agents do not know the
structure of the graph, but have some prior regarding it. The latter
is particularly compelling, since the assumption that the agents know
exactly the structure of the graph is a strong one, especially in the
case of large graphs.

We believe that our results should extend to these cases, but chose
not to pursue their study, given the length and considerable
complexity of the argument presented here.

Although we show that there exist non-egalitarian graphs with
equilibria at which learning fails, we are far from characterizing
those graphs. For example: is there a simple geometric
characterization of the infinite graphs on which the agents learn with
probability one?

\bibliographystyle{abbrv}
\bibliography{all}
\pagebreak
\appendix

\section{Rooted graphs}
\label{app:rooted-graphs-proofs}
\begin{claim}
  \label{claim:metric}
  If $D([G,i],[G',i'])=0$ then $(G,i) \cong (G',i')$.
\end{claim}
\begin{proof}
  By the definition of $D(\cdot,\cdot)$, $D([G,i],[G',i'])=0$ implies
  that $B_r(G,i) \cong B_r(G',i')$ for all $r \in \N$. Hence for each
  $r$ there exists a (finite) graph isomorphism $h_r$ from the
  vertices of $B_r(G,i)$ to the vertices of $B_r(G',i')$. The goal is
  to construct the (potentially infinite) graph isomorphism between
  $(G,i)$ and $(G',i')$.

  For each $r \in \N$, the isomorphism $h_r$ can be restricted to an
  isomorphism between $B_1(G,i)$ and $B_1(G',i')$. Since $B_1(G,i)$ is
  finite, there are only finitely many possible isomorphisms between
  it and $B_1(G',i')$, and therefore at least one of them must appear
  infinitely often in $\{h_r\}_{r \geq 1}$. Hence let $\{h_{r,1}\}_{r
    \geq 1}$ be an infinite subsequence of $\{h_r\}_{r \geq 1}$ that
  consists of isomorphisms that are identical, when restricted to
  balls of radius one. By the same argument, there exists a
  sub-subsequence $\{h_{r,2}\}_{r \geq 2}$ that agrees on balls of
  radius two. Indeed, for any $n \in \N$ let $\{h_{r,n}\}_{r \geq n}$
  be a subsequence of $\{h_{r,n-1}\}_{r \geq n-1}$ that agrees on
  balls of radius $n$.

  In the diagonal sequence $\{h_{r,r}\}_{r \geq 1}$, $h_{r,r}$ agrees
  on balls of radius $r$ with all $h_{s,s}$ such that $s \geq r$. We
  can therefore define an isomorphism $h : (G,i) \to (G',i')$ by
  specifying that that $h(j) = h_{r,r}(j)$ for all $r \geq \dist(i,j)
  + 1$. $h$ is indeed an isomorphism, since if $(j,k) \in E$ then
  $(h(j),h(k)) = (h_r(j),h_r(k)) \in E'$, where $r =
  \max\{\dist(i,j),\dist(i,k)\}$.
\end{proof}

Let $\graphs(d,L) \subset \scg$ be the subspace of isomorphism classes
of $(d,L)$-egalitarian strongly connected rooted graphs.
\begin{theorem}
  \label{thm:cBLd-compact}
  $\graphs(d,L)$ is compact.
\end{theorem}
\begin{proof}
  Let $\{[G_n,i_n]\}_{n=1}^\infty$ be a sequence in
  $\graphs(d,L)$. Since the degrees are bounded, it follows that for
  fixed $r$, the number of possible balls $B_r(G_n,i_n)$ is finite,
  and therefore, by a standard diagonalization argument, there exists
  a subsequence that converges to some $[G,i]$. It remains to show
  that any such $[G,i]$ is $L$-locally-connected. This follows from the fact
  that for any edge $(k,j)$ in $(G,i)$, $B_L(G,j) \cong
  B_L(G_n,j_n)$ for some $n$; this ball must then include a path
  from $k$ back to $j$ of length at most $L$.
\end{proof}

Rather than prove Lemma~\ref{cor:compact} directly, we prove the
following more general theorem, which might be of independent
interest, as it extends the well understood notion of compactness in
undirected graphs to directed, strongly connected
graphs. Lemma~\ref{cor:compact} is an immediate consequence.
\begin{theorem}
  \label{thm:compact-graph}
  Let $\mathcal{S} \subseteq \scg$ have the property that if $[G,i]$
  is in $\mathcal{S}$, and if $j$ is another vertex in $G$, then
  $[G,j]$ is also in $\mathcal{S}$. Then $\mathcal{S}$ is precompact
  in $\scg$ if and only if $\mathcal{S} \subseteq \graphs(d,L)$ for
  some $d$ and $L$.
\end{theorem}
\begin{proof}
  By Theorem~\ref{thm:cBLd-compact} $\graphs(d,L)$ is compact. Hence
  $\mathcal{S} \subseteq \graphs(d,L)$ implies that $\mathcal{S}$ is
  precompact.

  To prove the other direction, consider first a sequence
  $\{[G_n,i_n]\}$ in $\mathcal{S}$ such that the degree of $i_n$ is at
  least $n$. Then clearly $\{[G_n,i_n]\}$ has no converging
  subsequence, since the degree of $i'$ in any limit $[G',i']$ would
  have to be larger than any $n$. It follows that $\mathcal{S}$ is not
  precompact.

  Finally, let there exist a sequence $\{[G_n,i_n]\}$ in $\mathcal{S}$
  with a sequence of edges $\{(i_n,j_n)\}$ in $G_n$ where the shortest
  path from $j_n$ back to $i_n$ is at least of length $n$. Assume that
  $[G',i']$ is a limit of a subsequence of this sequence. It follows
  that there exists a $j' \in \neigh{i'}$ such that the shortest path
  from $j'$ back to $i'$ is of length larger than any $n$, and so
  doesn't exist. Hence $G'$ is not strongly connected, $[G',i'] \not
  \in \scg$, and $\mathcal{S}$ is not compact in $\scg$.
\end{proof}

The following is a general claim that will be useful later.
\begin{claim}
  \label{thm:graph-lims}
  Let $\{[G_n,i_n]\}_{n=1}^\infty$ be a sequence of rooted graph
  isomorphism classes such that
  \begin{align*}
    \lim_{n\to \infty}[G_n,i_n] = [G,i].
  \end{align*}
  Then for every $r >0$ there exists an $N > 0$ such that for all $n
  >N$ it holds that $B_r(G_n,i_n) \cong B_r(G,i)$. Furthermore, there
  exists a subsequence $\{[G_{n_r},i_{n_r}]\}_{r=1}^\infty$ such that
  $B_r(G_{n_r},i_{n_r}) \cong B_r(G,i)$.
\end{claim}
\begin{proof}
  The first part of the claim follows directly from the definition of
  limits and Eq.~\eqref{eq:graph-metric-def}. The second part holds for
  $n_r = \min\{n \,:\, B_r(G_n,i_n) \cong B_r(G,i)\}$, which is
  guaranteed to be finite by the first part.
\end{proof}

\section{Locality}
An important observation is that the actions and the utility of an
agent, up to {\em time} $t$, depends only on the strategies of the
agents that are at {\em distance} at most $t$ from it. We formalize
this notion in this section.

\begin{claim}
  Let $\cG_1=(\mu_0,\mu_1,\disc,G_1)$ and $\cG_2(\mu_0,\mu_1,\disc,G_2)$
  be games.  Let $h$ be a rooted graph isomorphism between
  $B_{r+1}(G_1,i_1)$ and $B_{r+1}(G_2,i_2)$ for some $r>0$, and let
  $\strat_1^{j_1} = \strat_2^{j_2}$ for all $j_1 \in
  B_r(G_1,i_1)$ and $j_2=h(j_1)$.

  Then the games, as probability spaces, can be coupled so that
  $\action^{i_1}_t=\action^{i_2}_t$ for all $t \leq r$.
\end{claim}
Some care needs to be taken with a statement such as ``agent $j_1$
plays the same strategy as agent $j_2$''; it can only be meaningful in
the context of a bijection that identifies each neighbor of $j_1$ with
each neighbor of $j_2$. We here naturally take this bijection to be
$h$, and accordingly demand that it be an isomorphism between balls of
radius $r+1$ (rather than $r$), so that the neighbors of the agents on
the surface of the ball are also mapped.
\begin{proof}
  Couple the two processes by equating the states of nature and
  setting $\psignal_{j_1} = \psignal_{h(j_1)}$ for all $j_1 \in
  B_r(G_1,i_1)$, and furthermore coupling the choices of pure
  strategies of $j_1$ and $h(j_1)$.

  We shall prove by induction a stronger statement, namely that under
  the claim hypothesis, $\action^{j_1}_t=\action^{j_2}_t$ for any $j_1
  \in B_r(G_1,i_1)$, $j_2=h(j_1)$ and $t \leq r - \dist(i_1,j_1)$.
  
  We prove the statement by induction on $t$. For $t=0$,
  $\action^{j_1}_0$ depends only on agent $j_1$'s private signal and
  choice of pure strategy, which are both equal to those of
  $j_2$. Hence $\action^{j_1}_0=\action^{j_2}_0$ for all $j \in
  B_r(G_1,i_1)$.

  Assume now that the claim holds up to some $t-1 \leq r-1$.  Let
  $j_1$ be such that $t \leq r - \dist(i_1,j_1)$. We would like to
  show that $\action^{j_1}_t=\action^{j_2}_t$. Let $k_1$ be a neighbor
  of $j_1$. Then $t-1 \leq r - \dist(i_1,k_1)$, and so
  $\action^{k_1}_{t'}=\action^{k_2}_{t'}$ for all $t' \leq t-1$, by the
  inductive assumption. Since $\action^{j_1}_t$ depends only on
  $j_1$'s private signals, choice of pure strategy and the actions of
  her neighbors in previous time periods, and since these are all
  identical to those of $j_2$, then it indeed follows that
  $\action^{j_1}_t=\action^{j_2}_t$.
\end{proof}

Recalling the definition
\begin{align*}
  \util_{i,t} = \P{\action^i_t=S},
\end{align*}
the following corollary is a direct consequence of this claim.
\begin{corollary}
  \label{thm:local-utils}
  Let $\cG_1=(\mu_0,\mu_1,\disc,G_1)$ and $\cG_2(\mu_0,\mu_1,\disc,G_2)$
  be games.  Let $h$ be a rooted graph isomorphism between
  $B_{r+1}(G_1,i_1)$ and $B_{r+1}(G_2,i_2)$ for some $r>0$, and let
  $\strat_1^{j_1} = \strat_2^{j_2}$ for all $j_1 \in
  B_r(G_1,i_1)$ and $j_2=h(j_1)$.

  Then $\util_{i_1,t}=\util_{i_2,t}$ for all $t \leq r$.
\end{corollary}

\section{A topology on strategies and the existence of equilibria
  for finite graphs}
\label{app:strategy-topology}
In the following theorem we show that the agents' set of strategies
admits a compact topology which preserves the continuity of the
utilities.  We use this topology to define our topology on equilibria,
which is a crucial component of the proof of our main theorem. We also
use it to infer the existence of equilibria for this game when the
number of players is infinite.

For a fixed private belief $\belief_i$, a pure strategy is a function
from the actions of neighbors to actions, which we call a
response. Formally, let $G=(V,E)$ be a social network.  A {\em
  response at time $t$} of an agent $i \in V$ is a function
$\resp_{i,t} : \{0,1\}^{|\neigh{i}|\cdot t} \to \{0,1\}$. A {\em
  response} of an agent $i$ is the sequence of functions
$\resp_i=(\resp_{i,0},\resp_{i,1},\ldots)$. Let $\cR_i$ be the space
of responses of agent $i$.

A (mixed) strategy of agent $i$ can be thought of as a measure on the
product space $[0,1] \times \cR_i$ of private beliefs and responses,
with the marginal on the first coordinate equaling the distribution of
$\belief_i$. Milgrom and Weber~\cite{milgrom1985distributional} call
this representation a {\em distributional strategy}. In the proof of
their Theorem 1, they show that for a game with incomplete information
and a finite number of players, and given some conditions, the weak
topology on distributional strategies is compact and keeps the
utilities continuous. Then, using Glicksberg's
theorem~\cite{glicksberg1952further} they infer that the game has an
equilibrium. The next theorem shows that these conditions apply in our
case, when the number of agents is finite.

\begin{lemma}
  \label{thm:finite-strat-topo}
  Fix $G=(V,E)$, with $V$ finite. Then for each agent $i$ there exists
  a topology $\cT_i$ on her strategy space such that the strategy
  space is compact and the utilities $\util_j$ are continuous in the
  product of the strategy spaces. Furthermore, there exists an
  equilibrium strategy profile.
\end{lemma}
\begin{proof}
  We prove by showing that the conditions of Theorem 1
  in~\cite{milgrom1985distributional} are met.
  \begin{enumerate}
  \item The set of private beliefs ({\em types} in the language
    of~\cite{milgrom1985distributional}) is $[0,1]$, a complete
    separable metric space, as required. Furthermore, the distribution
    of private beliefs is absolutely continuous with respect to the
    product of their marginal distributions. This fulfills condition
    R2 of~\cite{milgrom1985distributional}.
  \item The utilities $\util_j$ are bounded, measurable functions of
    the private beliefs and the responses.
  \item Define a metric $D$ on $i$'s responses $\cR_i$ by
    \begin{align*}
      D(\resp_i,\resp_i') = \exp\left(-\min\{t\,:\, \resp_{i,t} \neq
        \resp_{i,t}'\}\right).
    \end{align*}
    This can be easily verified to indeed be a metric. By a standard
    diagonalization argument it follows that $\cR_i$ is compact in the
    topology induced by this metric, as required.

    Furthermore, for fixed private beliefs, the utilities $\util_j$
    are equicontinuous in the responses: if a response is changed by
    at most $\delta = e^{-T}$ (in terms of the metric $D$) then it
    remains unchanged in the first $T$ time periods, and so the
    utilities are changed by at most $\eps =
    (1-\disc)\sum_{t=T}^\infty\disc^t = \disc^T$. This fulfills
    condition R1 of~\cite{milgrom1985distributional}.
  \end{enumerate}
  Since these conditions are met, it follows by the proof of Theorem 1
  in~\cite{milgrom1985distributional} that the mixed strategies of
  agent $i$ are compact in the weak topology $\cT_i$, and that the
  utilities $\util_j$ are, under $\cT_i$, a continuous function of the
  strategies. Furthermore, and again by Milgrom and Weber's Theorem 1,
  this game also has an equilibrium.
\end{proof}

Note that under the above defined topology on $\cR_i$ the set of pure
strategies is separable, and so the topology $\cT_i$ on (mixed)
strategies is metrizable, e.g.\ with the L\'evy-Prokhorov
metric~\cite{billingsley1999convergence}.

The following variant of Lemma~\ref{thm:finite-strat-topo} will be
useful later.
\begin{lemma}
  \label{thm:finite-strat-topo-time-t}
  Fix $G=(V,E)$, with $V$ finite. Then for each agent $i$ there exists
  a topology $\cT_i$ on her strategy space such that the strategy
  space is compact and the utilities in each time $t$, $\util_{j,t}$,
  are continuous in the product of the strategy spaces.
\end{lemma}
\begin{proof}
  The proof is identical to the proof of
  Lemma~\ref{thm:finite-strat-topo} above, except that we let each
  agent's expected utility be given by $\util'_j = \util_{j,t}$;
  that is, we set the discount factor to be one at time $t$ and $0$
  otherwise.  Since in the proof above we required of the discount
  factors nothing more than to have a finite sum, the proof still
  applies, and the utilities (in this case $\util_{j,t}$), are
  continuous in the product of the strategy spaces.
\end{proof}

\section{A topology on rooted graph strategy profiles}
\label{app:topology-proofs}

Let the {\em utility map at time $t$} $\util_t : \gs \to \R$ be given
by
\begin{align*}
  \util_t([G,i,\stratp]) = \util_{i,t}(G,\stratp).
\end{align*}

\begin{lemma}
  \label{thm:utility-cont}
  The utility map $u: \gs \to \R$ is continuous.
\end{lemma}
\begin{proof}
  We will prove the claim by showing that $\util_t$ is continuous. The
  claim will follow because, by the bounded convergence theorem, if
  $f$ is a linear combination of the uniformly bounded maps
  $\{f_t\}_{t=0}^\infty$, with summable positive coefficients, then
  the continuity of all the maps $f_t$ implies the continuity of $f$.

  Let $[G_n,i_n,\stratp_n] \to_n [G,i,\stratp]$. We will show that
  $\util_t([G_n,i_n,\stratp_n]) \to_n \util_t(G,i,\stratp)$.

  Consider a sequence of games $\cG_n$ which are all played on the
  finite graph $B = B_{t+1}(G,i)$.  Since $[G_n,i_n,\stratp_n] \to_n
  [G,i,\stratp]$ then there exists some $N$ such that, for $n>N$, it
  holds that $D([G_n,i_n,\stratp_n], [G,i,\stratp]) <
  2^{-(t+1)}$. Hence, by the definition of $D(\cdot,\cdot)$, it holds
  that $B \cong B_{t+1}(G_n,i_n)$ for $n>N$. Denote by $h_n$ an
  isomorphism between the two balls that minimizes $\max_{j \in
    B_t(G,i)}d_{h_n}(\strat^j,\strat_n^{h_n(j)})$, as appears in the
  definition of $D(\cdot,\cdot)$.

  Let each agent $j$ in $B_t(G,i)$ play $\strat_n^{h_n(j)}$ in
  $\cG_n$, and let the rest of the agents in $B$ (i.e., those at
  distance $t+1$ from $i$) play arbitrary strategies. Denote by
  $\stratpx_n$ the strategy profile played by the agents at game
  $\cG_n$, and denote by $\stratpx$ the restriction of $\stratp$ to
  $B$. By Corollary~\ref{thm:local-utils},
  $\util_t([G_n,i_n,\stratp_n]) = \util_t([B,i,\stratpx_n])$ and
  $\util_t([G,i,\stratp]) = \util_t([B,i,\stratpx])$. Therefore it is
  left to show that $\util_t([B,i,\stratpx_n]) \to_n
  \util_t([B,i,\stratpx])$.

  Now, $D([G_n,i_n,\stratp_n], [G,i,\stratp]) \to 0$. Hence
  $d(\stratx^j,\stratx_n^j) \to 0$, and so the strategies of each
  agent $j$ in $B_t(G,i)$ converge to the strategy $\stratx^j=
  \strat^j$. Furthermore, the strategies in $B_t(G,i)$ converge
  uniformly, since there is only a finite number of them, and so the
  strategy profile converges in the product topology. It follows that
  $\util_{i_n,t}$, which by Lemma~\ref{thm:finite-strat-topo-time-t}
  is a continuous function of the strategies in $B_t(G,i)$, converges
  to $\util_{i,t}$.
\end{proof}

\begin{lemma}
  \label{thm:equi-conv-equi}
  The set of equilibrium rooted graph strategy profiles is closed.
\end{lemma}
\begin{proof}
  Let $\lim_n[G_n,i_n,\stratp_n] = [G,i,\stratp]$, and let each
  $\stratp_n$ be an equilibrium.

  Let $\stratpx$ be a strategy profile for the agents in $G$ such that
  $\stratpx^j = \stratp^j$ for all $j \neq i$. We will show that
  $\util_i(G,\stratp) \geq \util_i(G,\stratpx)$.

  Let $\stratpx_n$ be the strategy profile for agents on $G_n$ defined
  by $\stratpx_n^j = \stratp_n^{j_n}$ for $j_n \neq i_n$, and let
  $\stratpx_n^{i_n} = \stratpx^i$. Note that $[G_n,i_n,\stratpx_n]
  \to_n [G,i,\stratx]$.

  Since $\stratp_n$ is an equilibrium profile of $\cG_n$,
  \begin{align*}
    \util_{i_n}(G,\stratp_n) \geq \util_{i_n}(G,\stratpx_n).
  \end{align*}
  Taking the limit of both sides and substituting the definition of
  the utility map we get that
  \begin{align*}
    \lim_{n \to \infty}\util([G_n,i_n,\stratp_n]) \geq \lim_{n \to
      \infty}\util(G_n,i_n,\stratpx_n).
  \end{align*}
  Finally, since by Lemma~\ref{thm:utility-cont} above the utility map
  is continuous, we have that
  \begin{align*}
    \util([G,i,\stratp]) \geq \util([G,i,\stratpx_n]).
  \end{align*}
\end{proof}

\begin{claim}
  \label{thm:conv-conv}
  Let $R \subseteq \scg$ be a compact set of rooted graphs, and let
  $\mathcal{SP}(R)$ be the set of rooted graph strategy profiles
  $[G,i,\stratp]$ with $[G,i] \in R$. Then $\mathcal{SP}(R)$ is
  compact.
\end{claim}
\begin{proof}
  Let $\{[G_n,i_n,\stratp_n]\}_{n=1}^\infty$ be a sequence of rooted
  graph strategy profiles in $\mathcal{SP}(R)$. We will prove the
  claim by showing that it has a converging subsequence.
  
  Let $[G,i]$ be the limit of some subsequence of
  $\{[G_n,i_n]\}_{n=1}^\infty$; this exists because $R$ is compact.
  
  By Claim~\ref{thm:graph-lims} there exists a sub-subsequence
  $\{[G_{n_r},i_{n_r}]\}_{r=1}^\infty$ such that $B_r(G_{n_r},i_{n_r})
  \cong B_r(G,i)$. We will therefore assume without loss of generality
  that $n_r=n$, i.e., limit ourselves to this sub-subsequence.
  Accordingly, let $h_n:V \to V_n$ be a sequence rooted graph
  isomorphisms between $B_n(G,i)$ and $B_n(G_n,i_n)$. Note that since
  out-degrees are finite then $B_n(G,i)$ is finite for all $n$.

  Let $j$ be a vertex of $G$, and let $r_j$ be the graph distance
  between $i$ and $j$. For $n \geq r_j+1$, denote
  $j_n=h_n(j)$. Note that $h_n$ also maps the neighbors of $j_n$
  to the neighbors of $j$.

  We will now construct $\stratp$, the strategy profile of the agents
  in $G$ such that $[G_{n_k},i_{n_k},\stratp_{n_k}] \to_k
  [G,i,\stratp]$, for some subsequence $\{n_k\}$.  We start with agent
  $1$ of $G$. Since $\cT_1$ is compact, the sequence
  $\{\strat_n^{1_n}\}_{n=r_1+1}^\infty$ has a converging subsequence,
  i.e., one along which $d_{h_n}(\strat_n^{1_n},\strat^1) \to_n 0$ for
  some strategy $\strat^1$, which we will assign to agent $1$ in
  $G$. Likewise, this subsequence has a subsequence along which
  $d_{h_n}(\strat_n^{2_n},\strat^2) \to_n 0$, etc. Thus, by a standard
  diagonalization argument, we have that there exists a subsequence
  $\{[G_{n_k},i_{n_k},\stratp_{n_k}]\}_{k=1}^\infty$ with isomorphisms
  $h_{n_k}$ such that $d_{h_{n_k}}(\strat_{n_k}^{j_{n_k}},\strat_j)
  \to_k 0$ for all $j$. It is now straightforward to verify that
  $D\Big([G_{n_k},i_{n_k},\stratp_{n_k}],[G,i,\stratp]\Big) \to_k 0$:
  pick some $r>0$ and then $k$ large enough so that $h_{n_k}$ is an
  isomorphism between $B_{r+1}(G_n,i_n)$ and $B_{r+1}(G,i)$. Then by
  definition (Eq.~\eqref{eq:D-def})
  \begin{align*}
    D\Big([G_{n_k},i_{n_k},\stratp_{n_k}],[G,i,\stratp]\Big) \leq
    \max\left\{2^{-r}, \max_{j \in
        B_r(G,i)}d_{h_{n_k}}(\strat_{n_k}^{j_{n_k}},\strat^j)\right\}.
  \end{align*}
  If we now further increase $k$ then
  $d_{h_{n_k}}(\strat_{n_k}^{j_{n_k}},\strat^j) \to_k 0$, and since
  $B_r(G,i)$ is finite then we have that
  $D\Big([G_{n_k},i_{n_k},\stratp_{n_k}],[G,i,\stratp]\Big) \leq
  2^{-r}$, for $k$ large enough. Since this holds for all $r$ then
  \begin{align*}
    D\Big([G_{n_k},i_{n_k},\stratp_{n_k}],[G,i,\stratp]\Big) \to_k 0
  \end{align*}
  and
  \begin{align*}
    \lim_{k \to \infty}[G_{n_k},i_{n_k},\stratp_{n_k}] = [G,i,\stratp].
  \end{align*}
\end{proof}

An easy consequence of Claim~\ref{thm:conv-conv} and
Lemma~\ref{thm:equi-conv-equi} is the following claim. Let $\cR$ be a
set of rooted graph isomorphism classes. Denote by $\eqs(\cR)$ the set
of rooted graph strategy profiles $[G,i,\stratp]$ such that $[G,i] \in
\cR$ and $\stratp$ is an equilibrium strategy profile. For a set of
graphs $\cK$, let $\eqs(\cK)$ be a shortened notation for
$\eqs(\cR(\cK))$.
\begin{claim}
  \label{clm:cB-compact}
  Let $\cR \in \scg$ be a precompact set of strongly connected rooted
  graphs. Then $\overline{\eqs(\cR)}$ is a compact set of equilibrium
  rooted graph strategy profiles.
\end{claim}
\begin{proof}
  We will prove the claim by showing that any sequence in $\eqs(\cR)$
  has a converging subsequence with a limit that is an equilibrium.

  Let $\{[G_n,i_n,\stratp_n]\}_{n=1}^\infty$ be a sequence of points
  in $\eqs(\cR)$. Since $\cR$ is precompact in $\scg$, the sequence
  $\{[G_n,i_n]\}_{n=1}^\infty$ has a converging subsequence
  $\{[G_{n_k},i_{n_k}]\}_{k=1}^\infty$ that converges to some $[G,i]
  \in \scg$. Hence, by Claim~\ref{thm:conv-conv}, the sequence
  $\{[G_n,i_n,\stratp_n]\}_{k=1}^\infty$ has a converging subsequence
  that, for some $\stratp$, converges to $[G,i,\stratp]$. Finally, by
  Lemma~\ref{thm:equi-conv-equi} $\stratp$ is an equilibrium strategy
  profile for $G$, and so $[G,i,\stratp]$ is an equilibrium rooted
  graph strategy.
\end{proof}

Using these claims, we are now ready to easily prove that every game
has an equilibrium; the one additional ingredient will be
Lemma~\ref{thm:finite-strat-topo}, which shows that finite games have
equilibria.
\begin{theorem}
  \label{thm:equi-exists}
  Every game $\cG$ has an equilibrium.  
\end{theorem}
\begin{proof}
  Let $\cG = (\mu_0,\mu_1,\disc,G)$. Let $i$ be a vertex in $G$,
  denote $G_n = B_n(G,i)$, and denote its root by $i_n$. Let
  $\left\{\cG_n = (\mu_0,\mu_1,\disc,G_n)\right\}_{n=1}^\infty$ be a
  sequence of finite games with equilibria strategy profiles
  $\stratp_n$; finite games have equilibria by
  Lemma~\ref{thm:finite-strat-topo}. Then $[G_n,i_n] \to_n [G,i]$, and
  so by Claim~\ref{thm:conv-conv} we have that there exists a strategy
  profile $\stratp$ and a subsequence
  $\{[G_{n_k},i_{n_k}]\}_{k=1}^\infty$ such that
  \begin{align*}
    \lim_{k \to \infty}[G_{n_k},i_{n_k},\stratp_{n_k}] = [G,i,\stratp].
  \end{align*}
  Finally, by Lemma~\ref{thm:equi-conv-equi}, $\stratp$ is an
  equilibrium profile of $\cG$.
\end{proof}

\section{Agreement}
\label{app:agreement-proofs}
Denote by $Z^i_t$ the log-likelihood ratio of the events $S=1$ and
$S=0$, conditioned on $\info^i_{t}$, the information available to
agent $i$ at time $t$
\begin{align*}
  Z^i_t = \log \frac{\CondP{S=1}{\info^i_t}}{\CondP{S=0}{\info^i_t}},
\end{align*}
and let
\begin{align*}
  Z^i_{\infty} = \log
  \frac{\CondP{S=1}{\info^i_\infty}}{\CondP{S=0}{\info^i_\infty}}.
\end{align*}
  
Let $Y^i_t$ be defined as follows:
\begin{align*}
  Y^i_t = \log \frac
  {\CondP{\hist{i}{t}}{\action^i_{[0,t)},\strat^i,S=1}}
  {\CondP{\hist{i}{t}}{\action^i_{[0,t)},\strat^i,S=0}},
\end{align*}
where $\action^i_{[0,t)}$ is the sequence of actions of $i$ up to time
$t-1$, and $\action^{\neigh{i}}_{[0,t)}$ is the sequence of actions of
$i$'s neighbors up to time $t-1$.  Finally, let
\begin{align*}
  Y^i_{\infty} = \lim_{t \to \infty}Y^i_t.
\end{align*}
Note that it is not clear that the limit $\lim_t Y^i_t$ exists. We show
this in the following claim.

\begin{claim}
  \label{thm:Y-Z}
  Denote by $\belief_{-i}$ the private beliefs of all agents but
  $i$. Then
  \begin{enumerate}
  \item $\lim_t Z^i_t=Z^i_{\infty}$ almost surely.
  \item $Z^i_t = Y^i_t + Z^i_0$.
  \item $\lim_t Y^i_t=Y^i_{\infty}$ almost surely, and
    $Y^i_{\infty}$ is measurable in
    $\sigma(\action^i_{[0,\infty)},\belief_{-i},\stratp)$.
  \end{enumerate}

\end{claim}
\begin{proof}
  \begin{enumerate}
  \item
    Recall that
    \begin{align*}
      Z^i_t &=  \log
      \frac
      {\CondP{S=1}{\info^i_t}}
      {\CondP{S=0}{\info^i_t}}.
    \end{align*}
    Since $\{\info^i_t\}_{t=0}^\infty$ is a filtration then
    ${\CondP{S=1}{\info^i_t}}$ is a martingale, which converges a.s.\ since
    it is bounded. Hence $Z^i_t$ also converges, and in particular
    \begin{align*}
      \lim_{t \to \infty} Z^i_t =  \log
      \frac
      {\CondP{S=1}{\info^i_{\infty}}}
      {\CondP{S=0}{\info^i_{\infty}}} = Z^i_{\infty}.
    \end{align*}
  \item 
    By the definition of $\info^i_t$
    \begin{align*}
      Z^i_t &=  \log
      \frac
      {\CondP{S=1}{\belief_i,\strat^i,\hist{i}{t}}}
      {\CondP{S=0}{\belief_i,\strat^i,\hist{i}{t}}}\\
      &= \log
      \frac
      {\CondP{\hist{i}{t}}{\belief_i,\strat^i,S=1}}
      {\CondP{\hist{i}{t}}{\belief_i,\strat^i,S=0}}
      \frac
      {\CondP{S=1}{\belief_i,\strat^i}}
      {\CondP{S=0}{\belief_i,\strat^i}},
    \end{align*}
    where the second equality follows from Bayes' law. Now, conditioned
    on $S$ and $i$'s pure strategy $\strat^i$, the probability for a
    sequence of actions $\hist{i}{t}$ of $i$'s neighbors depends on
    $\belief_i$ only in as much as $\belief_i$ affects $i$'s actions up
    to time $t-1$, $\action^i_{[0,t)}$. Hence
    $\CondP{\hist{i}{t}}{\belief_i,\strat^i,S} =
    \CondP{\hist{i}{t}}{\action^i_{[0,t)},\strat^i,S}$. Note also that
    \begin{align*}
      \frac
      {\CondP{S=1}{\belief_i,\strat^i}}
      {\CondP{S=0}{\belief_i,\strat^i}} = Z^i_0.
    \end{align*}
    Therefore
    \begin{align*}
      Z^i_t &=  Y^i_t+ Z^i_0.
    \end{align*}
  \item Since $Z^i_t$ converges almost surely and $Z^i_t = Y^i_t+
    Z^i_0$ then $Y^i_t$ also converges almost surely.  Since each
    $Y^i_t$ is a function of $\hist{i}{t}$ and $\strat^i$, it follows that
    their limit, $Y^i_{\infty}$, is measurable in
    $\sigma(\hist{i}{\infty},\strat^i)$. However, given $\stratp$,
    $\hist{i}{\infty}$ is a function of $\belief_{-i}$ and
    $\action^i_{[0,\infty)}$: for a choice of pure strategies the
    actions of all agents but $i$ can be determined given their
    private signals and the actions of $i$. Hence $Y^i_{\infty}$ is
    also measurable in $\sigma(\action^i_{[0,\infty)},\belief_{-i},\stratp)$.
  \end{enumerate}
\end{proof}

\begin{claim}
  \label{thm:Z-non-atomic}
  The distribution of $Z^i_0$ is non-atomic, as is the distribution of
  $Z^i_0$ conditioned on $S$.
\end{claim}
\begin{proof}
  By definition,
  \begin{align*}
    Z^i_0 = \log
    \frac
    {\CondP{S=1}{\belief_i,\strat^i}}
    {\CondP{S=0}{\belief_i,\strat^i}}.
  \end{align*}
  However, the choice of strategy $\strat^i$ is independent of both
  $\belief_i$ and $S$, and so
  \begin{align*}
    Z^i_0 = \log
    \frac
    {\CondP{S=1}{\belief_i}}
    {\CondP{S=0}{\belief_i}} = \log\frac{\belief_i}{1-\belief_i}.
  \end{align*}
  Since the distribution of $\belief_i$ is non-atomic (see the
  definition of $\belief_i$ and the comment after it) then so is the
  distribution of $Z^i_0$. Since $S$ takes only two values then the
  same holds when conditioned on $S$.
\end{proof}

\begin{theorem}
  \label{thm:bestres-optset}
  For any agent $i$ it holds that $\bestres^i_{\infty} =
  \optset_i$ almost surely, in any equilibrium.
\end{theorem}
\begin{proof}
  By its definition, $\optset_i$ takes values in
  $\{0,1,\{0,1\}\}$, and by Theorem~\ref{thm:rsv1} we have
  that $\optset_i \subseteq \bestres^i_{\infty}$. Therefore the claim
  holds when $\bestres^i_{\infty} = 0$ or $\bestres^i_{\infty} =
  1$, and it remains to show that $\optset_i = \{0,1\}$ when
  $\bestres^i_{\infty} = \{0,1\}$, or that $\P{\optset_i \neq
    \{0,1\},\bestres^i_{\infty}=\{0,1\}} = 0$.

  Let $a = (a_0,a_1,\ldots)$ be a sequence of actions, each in
  $\{0,1\}$, in which only one action appears infinitely often.  Since
  there are only countably many such sequences, then if $\P{\optset_i
    \neq \{0,1\},\bestres^i_{\infty}=\{0,1\}} > 0$, then there exists such a
  sequence $a$ for which $\P{\action^i_{[0,\infty)} = a,
    \bestres^i_{\infty}=\{0,1\}} > 0$. We shall prove the claim by showing that
  $\P{\action^i_{[0,\infty)} = a, \bestres^i_{\infty}=\{0,1\}} = 0$.

  Recall that by Claim~\ref{thm:Y-Z}, the event that
  $\bestres^i_{\infty}=\{0,1\}$ is equal to the event that $Z^i_0 =
  -Y^i_{\infty}$. Recall also that by the same claim, $Y^i_{\infty}$
  is measurable in
  $\sigma(\action^i_{[0,\infty)},\belief_{-i},\stratp)$. Hence
  \begin{align*}
    \CondP{\action^i_{[0,\infty)}=a, \bestres^i_{\infty} = \{0,1\}}{S,\belief_{-i},\stratp} &=
    \CondP{\action^i_{[0,\infty)}=a,Z^i_0 =
      -Y^i_{\infty}(a,\belief_{-i},\stratp)}{S,\belief_{-i},\stratp}\\
    &\leq \CondP{Z^i_0 =
      -Y^i_{\infty}(a,\belief_{-i},\stratp)}{S,\belief_{-i},\stratp}
  \end{align*}
  Now, by Claim~\ref{thm:Z-non-atomic}, $Z^i_0$ conditioned on $S$
  has a non-atomic distribution. Further conditioning on $\stratp$ and
  $\belief_{-i}$ leaves its distribution unchanged, since it is
  independent of the former, and independent of the latter conditioned
  on $S$. Hence the probability that it equals
  $-Y^i_{\infty}(a,S,\belief_{-i},\stratp)$ is $0$. Hence
  \begin{align*}
    \P{\action^i_{[0,\infty)}=a, Z^i_0} &=
    \E{\CondP{\action^i_{[0,\infty)}=a, Z^i_0 =
      -Y^i_{\infty}}{S,\belief_{-i},\stratp}}  = 0.
  \end{align*}
\end{proof}

\begin{proof}[Proof of Theorem~\ref{thm:agreement}]
  Let $i$ and $j$ be agents. Since $G$ is strongly connected, there
  exists a path from $i$ to $j$. By Theorem~\ref{thm:rsv-agreement} we
  have, by induction along this path, that $\optset_j \subseteq
  \bestres^i_{\infty}$ almost surely. But $\optset_i =
  \bestres^i_{\infty}$ by Theorem~\ref{thm:bestres-optset} above, and
  so we have that $\optset_j \subseteq \optset_i$. However, there also
  exists a path from $j$ back to $i$, and so $\optset_i \subseteq
  \optset_j$, and the two are equal. This holds for any pair of
  agents, and so it follows that there exists a random variable
  $\optset$ such that $\optset_i=\optset$ for all $i$, almost surely.
\end{proof}

\section{$\delta$-independence}
\label{app:delta-ind}

In this section we introduce a technical notion of {\em
  $\delta$-independent} random variables and {\em $(p,\delta)$-good
  estimators}\footnote{ These definitions are taken (almost) verbatim
  from~\cite{mossel2012asymptotic}.}.  These definitions will be
useful in the proof of our main theorem. Informally,
$\delta$-independent random variables are almost independent. The
random variables $(X_1,\ldots,X_n)$ are $(p,\delta)$-good estimators
of $S$ if each is equal to $S$ with probability at least $p$, and if
they are $\delta$-independent, conditioned on both $S=0$ and $S=1$.

Let $\dtv(\cdot,\cdot)$ denote the total variation distance between
two measures defined on the same measurable space, and let
$(X_1,X_2,\dots,X_k)$ be random variables.  We refer to them as {\em
  $\delta$-independent} if
\begin{align*}
  \dtv(\mu_{(X_1,\ldots,X_k)}, \mu_{X_1}\times\cdots\times\mu_{X_k}) \leq \delta,
\end{align*}
where $\mu_{(X_1,\ldots,X_k)}$ is their joint distribution, and
$\mu_{X_1}\times\cdots\times\mu_{X_k}$ is the product of their
marginal distributions.  I.e., the joint distribution
$\mu_{(X_1,\ldots,X_k)}$ has total variation distance of at most
$\delta$ from the product of the marginal distributions
$\mu_{X_1}\times\cdots\times\mu_{X_k}$.  Likewise, $(X_1,\ldots,X_l)$
are {\em $\delta$-dependent} if the total variation distance between
these distributions is more than $\delta$.

Let $S$ be a binary random variable such that $\P{S=1}=1/2$. We say
that the binary random variables $(X_1,\ldots,X_k)$ are
$(p,\delta)$-good estimators of $S$ if they are $\delta$-independent
conditioned both on $S=0$ and on $S=1$, and if $\P{X_\ell = S} \geq
p$, for $\ell =1,\ldots,k$.

The following standard concentration of measure lemma captures the
idea that the aggregation of sufficiently many $(p,\delta)$-good
estimators gives an arbitrarily good estimate, for any $p>\half$ and
for $\delta$ small enough. 
\begin{claim}
  \label{clm:k-almost-indep-ests}
  Let $S$ be a binary random variable such that $\P{S=1}=1/2$, and let
  $(X_1,\ldots,X_k)$ be $(\half+\eps,\delta)$-good estimators of $S$.
  Then there exists a function $a : \{0,1\}^k \to \{0,1\}$ such that
  \begin{align*}
   \P{a(X_1,\ldots,X_k) = S} > 1-e^{-2\eps^2k}-\delta. 
  \end{align*}  
\end{claim}
\begin{proof}
  Let $(Y_1,\ldots,Y_k)$ be random variables such that the
  distribution of $(S,Y_i)$ is equal to the distribution of $(S,X_i)$
  for all $i$, and let $(Y_1,\ldots, Y_k)$ be independent, conditioned
  on $S$. Then $(X_1,\ldots,X_k)$ can be coupled to $(Y_1,\ldots,Y_k)$
  in such a way that they differ only with probability
  $\delta$. Therefore, if we show that $\P{a(Y_1, \ldots, Y_k)=S} >
  q+\delta$ for some $a$ then it will follow that $\P{a(X_1, \ldots,
    X_k)=S} > q$.
  
  Denote $\hat{Y} = \frac{1}{k}\sum_{i=1}^k Y_i$, and denote $\alpha_0
  = \CondE{\hat{Y}}{S=0}$ and $\alpha_1
  = \CondE{\hat{Y}}{S=1}$ . It follows that
  \begin{align*}
    \alpha_1 - \alpha_0
    &= \frac{1}{k}\sum_{i=1}^k\left(2\P{Y_i=S}-1\right) > 2\eps.
  \end{align*}
  By the Hoeffding bound
  \begin{align*}
    \CondP{\hat{Y} \leq \alpha_1 - \eps}{S=1} < e^{-2\eps^2k}
  \end{align*}
  and
  \begin{align*}
    \CondP{\hat{Y} \geq \alpha_0 + \eps}{S=0} < e^{-2\eps^2k}.
  \end{align*}
  Let $a(Y_1,\ldots,Y_k) = \ind{\hat{Y}>\alpha_1-\eps}$. Then by the
  above we have that $\P{a(Y_1,\ldots,Y_k) \neq S} <
  e^{-2\eps^2k}$, and so
  \begin{align*}
    \P{a(X_1,\ldots,X_k) = S} > 1- e^{-2\eps^2k} -\delta.
  \end{align*}

\end{proof}

\section{The probability of learning}
\label{app:learnprob}

In this section we start to explore the probability of learning, with
the ultimate goal of proving that it equals $1$ under the appropriate
conditions.

Recall that the {\em probability of learning} map $\learnprob : \gs
\to \R$ is given by
\begin{align*}
  \learnprob([G,i,\stratp]) = \lim_{t \to
    \infty}\P{\action^i_t(G,\stratp)=S}.
\end{align*}

Before showing that $\learnprob$ is well defined (i.e., the limit
exists), and proving that it is lower semi-continuous, we make the
following additional definition. Let
\begin{align*}
  \mapS_{\infty} = \mapS_{\infty}([G,i,\stratp])=
  \begin{cases}
    0&\bestres^i_{\infty}(G,\stratp) = 0\\
    1&\bestres^i_{\infty}(G,\stratp) = 1 \,\mbox{ or
    }\,\bestres^i_{\infty}(G,\stratp) = \{0,1\}
  \end{cases}
\end{align*}
be a maximum a posteriori (MAP) estimator of $S$ given
$\info^i_{\infty}$, for some agent $i$ in $G$.

Note that $\mapS_{\infty}([G,i,\stratp])$ is independent of $i$,
since, by Theorem~\ref{thm:agreement}, $\bestres^i_{\infty} = \optset$
for all agents $i,j$ in $G$.  Note also that $\mapS_{\infty}$ is
indeed a MAP estimator of $S$ given $\info^i_{\infty}$, since by
definition $\bestres^i_{\infty}$ is the set of most probable estimates
of $S$, given $\info^i_{\infty}$.

\begin{claim}
  \label{clm:learnprob-bestres}
  \begin{align*}
    \learnprob([G,i,\stratp]) =  \P{\mapS_{\infty}([G,i,\stratp])=S}.
  \end{align*}
\end{claim}
\begin{proof}
  We first condition on the event\footnote{Note that this might be a
    zero probability event (which we are not able to prove or
    disprove), in which case the following argument is not needed.}
  that $\optset = \{0,1\}$. Then
  \begin{align*}
    \lim_{t \to
      \infty}\CondP{\action^i_t=S}{\optset=\{0,1\}} =
    \lim_{t \to
      \infty}\CondP{\action^i_t=S}{\bestres^i_{\infty}=\{0,1\}}.
  \end{align*}
  Since the event that $\bestres^i_{\infty}=\{0,1\}$ is identical to the event
  that $\CondP{S=1}{\info^i_{\infty}} = \half$, and since
  $\action^i_t$ is $\info^i_{\infty}$-measurable for all $t$, then
  it follows that
  \begin{align*}
    \lim_{t \to
      \infty}\CondP{\action^i_t=S}{\optset=\{0,1\}} = \half.
  \end{align*}
  and also that
  \begin{align*}
    \lim_{t \to
      \infty}\CondP{\mapS_{\infty}=S}{\optset=\{0,1\}} = \half.
  \end{align*}
  When $\optset=0$ or $\optset=1$ then
  $\lim_t A_{i,t}=\mapS_{\infty}$, and so
  \begin{align*}
    \lim_{t \to \infty}\CondP{\action^i_t=S}{\optset \neq \{0,1\}} =
    \CondP{\mapS_{\infty}=S}{\optset\neq\{0,1\}}.
  \end{align*}
  Since we have equality when conditioning on both
  $\optset\neq\{0,1\}$ and $\optset=\{0,1\}$ then we also have
  unconditioned equality and
  \begin{align*}
    \learnprob([G,i,\stratp]) = \lim_{t \to \infty}\P{\action^i_t=S}
    = \P{\mapS_{\infty}=S}.
  \end{align*}
\end{proof}

It follows that $\learnprob$ is well defined.  Since
$\mapS_{\infty}([G,i,\stratp])$ is independent of $i$ then the
following is a direct consequence of
Claim~\ref{clm:learnprob-bestres}.
\begin{corollary}
  $\learnprob([G,i,\stratp]) = \learnprob([G,j,\stratp])$ for
  all $i,j$.
\end{corollary}

Another consequence is that if $\learnprob([G,i,\stratp])=1$ then
$\mapS_{\infty}=S$ almost surely. Since we could have also defined
$\mapS_{\infty}$ to equal $0$ when $\optset = \{0,1\}$, it follows
that also $\optset = S$ almost surely. Hence
$\learnprob([G,i,\stratp])=1$ if and only if the agents learn $S$:
\begin{claim}
  \label{clm:p-learning}
  $\learnprob([G,i,\stratp])=1$ if and only if $\lim_t\action^j_t=S$
  almost surely for all agents $j$ in $G$.
\end{claim}

We next show that, not surprisingly, if the private signals are
informative then $\learnprob > 1/2$.  Given $\mu_0$ and $\mu_1$,
denote $\pstar = \half + \half \dtv(\mu_0,\mu_1)$.
\begin{claim}
  \label{clm:p-half}
  Given $\mu_0$ and $\mu_1$
  \begin{align*}
    \learnprob([G,i,\stratp]) \geq \pstar > \half
  \end{align*}
  for any $G$, $i$ and equilibrium strategy profile $\stratp$.
\end{claim}
\begin{proof}
  $\pstar > \half$, since $\mu_0 \neq \mu_1$.  Let $\hat{S}_{i,0}$
  be the maximum a posteriori (MAP) estimator of $S$ given $i$'s
  private signal, $\psignal_i$. Then (see Claim 3.30
  in~\cite{mossel2012asymptotic})
  \begin{align*}
    \P{\hat{S}_{i,0}=S} = \pstar.
  \end{align*}
  Now, $\mapS_{\infty}$ is a MAP estimator of $S$ given
  $\info^i_{\infty}$. Since $\info^i_{\infty}$ includes $\psignal_i$
  then
  \begin{align*}
    \P{\mapS_{\infty} = S} \geq \P{\hat{S}_{i,0}=S} = \pstar,
  \end{align*}
  and the claim follows by Claim~\ref{clm:learnprob-bestres}.
\end{proof}
We end this section with our main claim regarding $p$.
\begin{theorem}
  \label{thm:p-semi-cont}
  $\learnprob$ is lower semi-continuous, i.e., if
  $[G_n,i_n,\stratp_n] \to_n [G,i,\stratp]$ then
  \begin{align*}
   \liminf_n
   \learnprob([G_n,i_n,\stratp_n]) \geq \learnprob([G,i,\stratp]). 
  \end{align*}
\end{theorem}
\begin{proof}
  Recall that the expected utility of agent $i$ at time $t$ is given by the
  utility map at time $t$:
  \begin{align*}
    \util_t([G,i,\stratp]) = \P{\action^i_t(G,\stratp) = S}.
  \end{align*}
  Hence an alternative definition of $\learnprob$ is that
  \begin{align}
    \label{eq:learnprob-lim}
    \learnprob([G,i,\stratp]) = \lim_{t \to \infty}\util_t([G,i,\stratp]).
  \end{align}
  Now, $\action^i_t$ is $\info^i_\infty$-measurable. Hence, since
  $\mapS_{\infty}$ is a MAP estimator of $S$ given $\info^i_{\infty}$, it
  follows that
  \begin{align*}
    \P{\mapS_{\infty} = S} \geq \P{\action^i_t = S},
  \end{align*}
  or that
  \begin{align*}
    \learnprob([G,i,\stratp]) \geq \util_t([G,i,\stratp]).
  \end{align*}
  This, combined with~\eqref{eq:learnprob-lim}, yields
  \begin{align*}
    \learnprob([G,i,\stratp]) = \sup_t\util_t([G,i,\stratp]),
  \end{align*}
  and since $\util_t$ is continuous (see the proof of
  Lemma~\ref{thm:utility-cont}), it follows that $\learnprob$ is lower
  semi-continuous.
\end{proof}

\section{Finding good estimators}
\label{app:ind-ests-proof}

The following lemma is the technical core of the proof of the main
result of this article. Before stating it we would like to remind the
reader that if $\cK$ is a set of graphs then $\cR(\cK)$ is the set of
rooted $\cK$ graphs, and $\eqs(\cR(\cK))$ is the set of $\cR(\cK)$
equilibrium graph strategy profiles. Recall also that $\pstar = \half
+ \half\dtv(\mu_0,\mu_1)$.

\begin{lemma}
  \label{lem:independent-ests}
  Let $G$ be an infinite, strongly connected graph such that
  $\overline{\eqs(\cR(\{G\}))}$ is a compact set of equilibrium rooted
  graph strategy profiles. Then for all equilibrium strategy profiles
  on $G$, $k \in \N$, $\eps>0$ and $\delta>0$ there exists an agent
  $i$ in $G$, a time $t$ and $\info^i_t$-measurable binary random
  variables $(X_1,\ldots,X_k)$ that are $(\pstar-\eps,\delta)$-good
  estimators of $S$.
\end{lemma}
Before proving this lemma we will need some additional definitions and
claims.

We shall make use of the following notation: Let
$X_1,\ldots,X_k$ be random variables, and let $S$ be a binary random
variable. We say that $(X_1,\ldots,X_k)$ are $\delta$-independent
conditioned on $S$ if they are $\delta$-independent conditioned on
both $S=0$ and $S=1$. Denote
\begin{align*}
  \infdep_S(X_1,\ldots,X_k) =
  \min\{\delta\,:\,(X_1,\ldots,X_k) \mbox{ are
    $\delta$-independent conditioned on $S$}\}
\end{align*}
Note that this minimum is indeed attained, by the definition of
$\delta$-independence.

The proofs of the following two general claims are elementary and fairly
straightforward. They appear in~\cite{mossel2012asymptotic}.
\begin{claim}
  \label{clm:delta-independent}
  Let $A$, $B$ and $C$ be random variables such that $\P{A \neq B} \leq
  \delta$ and $(B,C)$ are $\delta'$-independent. Then $(A,C)$ are
  $2\delta+\delta'$-independent.
\end{claim}
\begin{claim}
  \label{clm:delta-ind-additive}
  Let $A=(A_1,\ldots,A_k)$, and $X$ be random variables. Let
  $(A_1,\ldots,A_k)$ be $\delta_1$-independent and let $(A,X)$ be
  $\delta_2$-independent. Then $(A_1,\ldots,A_k,X)$ are
  $(\delta_1+\delta_2)$-independent.
\end{claim}

\begin{claim}
  \label{clm:maps-local}
  Let $[G,i_0,\stratp]$ be an equilibrium graph strategy. Let
  $\{i_n\}_{n=1}^\infty$ be a sequence of vertices such that the graph
  distance $\dist(i_0,i_n)$ diverges with $n$. Fix $t$, and for each $n$
  let $X^{i_n}$ be $\info^{i_n}_t$-measurable. Then
  \begin{align*}
    \lim_{n \to \infty} \infdep_S(X^{i_n},\mapS_{\infty}) = 0.
  \end{align*}
\end{claim}
\begin{proof}
  Let $\mathcal{B}^i_r = \sigma(\{\psignal_j,\strat^j\,:\, j \in
  B_r(G,i)\})$. We first show, by induction on $r$, that $\info^i_r
  \subseteq \mathcal{B}^i_r$: any $\info^i_r$-measurable random
  variable is also $\mathcal{B}^i_r$-measurable. It will follow that
  $X^{i_n}$ is $\mathcal{B}^{i_n}_t$-measurable.

  Clearly $\info^i_0 \subseteq \mathcal{B}^i_0$. Assume now that
  $\info^j_{r'} \subseteq \mathcal{B}^j_{r'}$ for all $j$ and $r' <
  r$. By definition, $\info^i_r =
  \sigma(\info^i_{r-1},\action^{\neigh{i}}_{r-1})$. For $j
  \in \neigh{i}$ we have that $\action^j_{r-1}$ is
  $\mathcal{B}^j_{r-1}$-measurable. Finally, $\mathcal{B}^j_{r-1}
  \subseteq \mathcal{B}^i_r$, and so $\info^i_r
  \subseteq \mathcal{B}^i_r$.

  Note that for $i,j$ and $r_1,r_2$ such that $B_{r_1}(G,i)$ and
  $B_{r_2}(G,j)$ are disjoint it holds that $\mathcal{B}^i_{r_1}$ and
  $\mathcal{B}^j_{r_2}$ are independent conditioned on $S$, since the
  choices of pure strategies are independent and private beliefs
  are independent conditioned on $S$.

  Let $R^i_r$ be a MAP estimator of $\mapS_{\infty}$ given
  $\mathcal{B}^i_r$.  Since $\dist(i_0,i_n) \to_n \infty$, it follows
  that for any $r$ and $n$ large enough $B_t(G,i_n)$ and $B_r(G,i_0)$
  are disjoint, and so $X^{i_n}$ and $R^{i_0}_r$ are independent,
  conditioned on $S$. For such $n$, by
  Claim~\ref{clm:delta-independent}, we have that
  $(X^{i_n},\mapS_{\infty})$ are $2\delta$-independent, for $\delta =
  \P{R^i_r \neq \mapS_{\infty}}$.

  Finally, since $\mapS_{\infty}$ is
  $\mathcal{B}^i_\infty$-measurable, it follows that
  \begin{align*}
    \lim_{r \to \infty}\P{R^i_r \neq \mapS_{\infty}} = 0,
  \end{align*}
  and so 
  \begin{align*}
    \lim_{n \to \infty} \infdep_S(X^{i_n},\mapS_{\infty}) = 0.
  \end{align*}
\end{proof}

We are now ready to prove Lemma~\ref{lem:independent-ests}.
\begin{proof}[Proof of Lemma~\ref{lem:independent-ests}]
  Denote by $\mathcal{C}$ the closure of $\eqs(\cR(\{G\}))$. Note that
  by Lemma~\ref{thm:equi-conv-equi} any graph strategy in
  $\mathcal{C}$ is an equilibrium.

  We shall prove by induction a stronger claim, namely that under the
  claim hypothesis, for all $[H,j,\stratp] \in \mathcal{C}$, $k \in
  \N$, $\eps>0$ and $\delta>0$ there exists an agent $i$ in $H$, a
  time $t$ and $\info^i_t$-measurable binary random variables
  $(X_1,\ldots,X_k)$ that are $(\pstar-\eps,\delta)$-good estimators
  of $S$.

  We prove the claim by induction on $k$. The claim holds trivially
  for $k=0$. Assume that the claim holds up to $k$.

  Let $[H,j,\stratp] \in \mathcal{C}$.  Let $\{j_n\}_{n=1}^\infty$ be
  a sequence of vertices in $H$ such that
  $\lim_n\dist(j,j_n)=\infty$. Since $\mathcal{C}$ is compact then
  there exists a converging sequence $[H,j_n,\stratp] \to_n
  [F,i',\stratpx] \in \mathcal{C}$. By the inductive assumption, there
  exists an agent $i$ in $F$, a time $t$ and random variables
  $(X_1^i,\ldots,X_k^i)$ which are $\info^i_t$-measurable and are
  $(\pstar-\eps',\delta')$-good estimators of $S$, for some $0 < \eps'
  < \eps$ and $0 < \delta' < \delta$.  Denote
  \begin{align*}
   \bar{X}^i_k = (X^i_1,\ldots,X^i_k).
  \end{align*}
  Let $r = \dist(i',i)$. Since for $n$ large enough $B_r(H,j_n) \cong
  B_r(F,i')$ then, if we let $i_n \in B_r(H,j_n)$ be the vertex that
  corresponds to $i \in B_r(F,i')$ then $[H,i_n,\stratpx]
  \to_n [F,i,\stratpx]$, and $\lim_n\dist(j,i_n) = \infty$.

  Since $\bar{X}^i_k$ is $\info^i_t$-measurable, there exists a
  function $x^i_k$ such that $\bar{X}^i_k =
  x^i_k(\belief_i,\action^{\neigh{i}}_{[0,t)})$. Denote
  $\bar{X}^{i_n}_k  = x^i_k(\belief_{i_n},\action^{\neigh{i_n}}_{[0,t)})$.

  Now, since the strategies of agents in the neighborhood of $i_n$ in
  $H$ converge in the weak topology to those of $i$ in $F$, then the
  random variables $\left\{(S,\bar{X}^{i_n}_k)\right\}_{n=1}^\infty$
  converge in the weak topology to $(S,\bar{X}^i_k)$. Moreover, the
  measures of these random variables are over the finite space
  $\{0,1\}^{k+1}$, and so we also have convergence in total
  variation. In particular, $(X^{i_n}_1,\ldots,X^{i_n}_k)$ approach
  $\delta'$-independence:
  \begin{align}
    \label{eq:est-ind-converges}
    \lim_{n \to \infty} \infdep_S(X^{i_n}_1,\ldots,X^{i_n}_k) =
    \infdep_S(X^i_1,\ldots,X^i_k) \leq \delta'.
  \end{align}
  Likewise,
  \begin{align}
    \label{eq:est-prob-converges}
    \lim_{n \to \infty} \P{X^{i_n}_\ell=S} = \P{X^i_\ell=S} > \pstar - \eps'.
  \end{align}
  for $\ell=1,\ldots,k$. Additionally, since $\dist(j,i_n)
  \to_n \infty$, it follows by Claim~\ref{clm:maps-local} that
  \begin{align}
    \label{eq:est-ind-optset-converges}
    \lim_{n \to \infty} \infdep_S(\bar{X}^{i_n}_k,\mapS_{\infty}) = 0,
  \end{align}
  that is, $\bar{X}^{i_n}_k$ and $\mapS_{\infty}$ are practically
  independent, for large $n$.

  Now, recall that $\mapS_{\infty}$ is
  $\info^i_{\infty}$-measurable. Therefore, if we let $\estest^{i_n}_{t'}$
  be a MAP estimator of $\mapS_{\infty}$ given $\info^i_{t'}$ then
  for any $n$ it holds that
  \begin{align}
    \label{eq:estimating-estimator}
    \lim_{t' \to \infty}\P{\estest^{i_n}_{t'}=\mapS_{\infty}} = 1.
  \end{align}
  By Claim~\ref{clm:delta-independent}, a consequence of
  Eqs.~\eqref{eq:est-ind-optset-converges} and~\eqref{eq:estimating-estimator}
  is that
  \begin{align*}
    \lim_{n \to \infty} \lim_{t' \to \infty} \infdep_S(\bar{X}^{i_n}_k,\estest^{i_n}_{t'}) =
    0.
  \end{align*}
  That is, $\bar{X}^{i_n}_k$ and $\estest^{i_n}_t$ are practically
  independent, for large enough $n$ and $t'$. It follows by
  Claim~\ref{clm:delta-ind-additive} that
  \begin{align}
    \label{eq:limit-delta-ind}
    \lim_{n \to \infty} \lim_{t' \to \infty}
    \infdep_S(X^{i_n}_1,\ldots,X^{i_n}_k,\estest^{i_n}_{t'}) \leq \delta'.
  \end{align}

  It follows from Eq.~\eqref{eq:estimating-estimator} that
  \begin{align}
    \label{eq:estimating-estimator-good}
    \lim_{t' \to \infty}\P{\estest^{i_n}_{t'}=S} =
    \P{\mapS_{\infty}=S} \geq  \pstar.
  \end{align}

  Gathering the above results, we have that for $n$ and $t'$ large enough,
  \begin{enumerate}
  \item $\P{X^{i_n}_\ell=S} \geq \pstar-\eps$, by
    Eq.~\eqref{eq:est-prob-converges}. Likewise
    $\P{\estest^{i_n}_{t'}=S} = \P{\mapS_{\infty}=S} \geq \pstar-\eps$
    by Eq.~\eqref{eq:estimating-estimator-good}.
  \item $(X^{i_n}_1,\ldots,X^{i_n}_k,\estest^{i_n}_{t'})$ are
    $\delta$-independent, by Eq.~\eqref{eq:limit-delta-ind}.
  \end{enumerate}
  We therefore have that
  $(X^{i_n}_1,\ldots,X^{i_n}_k,\estest^{i_n}_{t'})$ are
  $\info^{i_n}_{t'}$-measurable $(\pstar-\eps,\delta)$-good estimators
  of $S$.  
\end{proof}

\section{Proof of main theorem}
The next theorem is a more general version of our main theorem, for
infinite graphs.
\begin{theorem}
  \label{thm:compact-learning}
  Let $G$ be an infinite, strongly connected graph such that
  $\overline{\eqs(\cR(\{G\}))}$ is a compact set of equilibrium rooted
  graph strategy profiles, and let $\stratp$ be any equilibrium strategy
  profile of $\cG$. Then
  \begin{align*}
    \P{\lim_t\action^i_t = S} = 1
  \end{align*}
  for all agents $i$.
\end{theorem}
It follows that $\P{\lim_t\action^i_t = S} = 1$ whenever $G$ is an
infinite, strongly connected, $(d,L)$-egalitarian graph, by
Theorem~\ref{thm:compact-graph} and Claim~\ref{clm:cB-compact}.

\begin{proof}
  We first show that $\learnprob = \learnprob([G,i,\stratp]) =
  1$. Assume by way of contradiction that $\learnprob < 1$. By
  Claim~\ref{clm:p-half}, we have that $p > 1/2$.

  By Lemma~\ref{lem:independent-ests}, for every $k \in N$, and
  $\delta>0$ there exist $(X_1,\ldots,X_k)$ that are
  $\info^i_t$-measurable for some $i$ and $t$, are
  $\delta$-independent conditioned on $S$, and are each equal to $S$
  with probability bounded away from one half, since $\pstar > \half$.

  By Claim~\ref{clm:k-almost-indep-ests} it follows that for $k$ large
  enough and $\delta$ small enough, there exists an estimator
  $\hat{S}$ of $S$ that is a function of $(X_1,\ldots,X_k)$, and is
  equal to $S$ with probability strictly greater than $\learnprob$.

  This $\hat{S}$ is $\info^i_{\infty}$-measurable, and so a MAP
  estimator of $S$ given $\info^i_{\infty}$ must also equal $S$ with
  probability greater than $\learnprob$. However, $\mapS_{\infty}$ in
  a MAP estimator of $S$ given $\info^i_{\infty}$, and it equals $S$
  with probability $\learnprob$ (Claim~\ref{clm:learnprob-bestres}),
  and so we have a contradiction. Hence $\learnprob([G,i,\stratp]) =
  1$.

  Now, by Claim~\ref{clm:learnprob-bestres} we have that 
  \begin{align*}
    \P{\mapS_{\infty}([G,i,\stratp])=S} = \learnprob([G,i,\stratp]) =  1.
  \end{align*}
  By the definition of $\mapS_{\infty}$ we have that $\mapS_{\infty} =
  S$ if and only if $\bestres^i_{\infty} = S$ for some (equivalently
  all) $i$. Since, by Theorem~\ref{thm:agreement}, $\optset =
  \optset_i = \bestres^i_{\infty}$, it follows that $\P{\optset = S} =
  1$, and that therefore $\P{\lim_t\action^i_t = S} = 1$.
\end{proof}

\pagebreak
\section*{Online Appendix: Examples}
\label{app:examples}
In this appendix we give two examples showing that the assumptions of
bounded out-degree and $L$-connectedness are crucial.  Our approach in
constructing equilibria will be to prescribe the initial moves of the
agents and then extend this to an equilibrium strategy profile.

Define the set of times and histories agents have to respond to as
$\responseSpace=\{(i,t,a):i\in V, t \in \mathbb{N}_0, a \in
[0,1]\times \{0,1\}^{|\neigh{i}|\cdot t}\}$.  The set
$[0,1]\times\{0,1\}^{|\neigh{i}|\cdot t}$ is interpreted as the pair
of the private belief of $i$ and the history of actions observed by
agent $i$ up to time $t$.  If
$a\in[0,1]\times\{0,1\}^{|\neigh{i}|\cdot t}$ then for $0\leq t' \leq
t$ we let $a_{t'}\in[0,1]\times\{0,1\}^{|\neigh{i}|\cdot t'}$ denote
the history restricted to times up to $t'$.  We say that a subset
$\responseSubset\subseteq\responseSpace$ is \emph{history-closed} if for
every $(i,t,a)\in\responseSubset$ we have that for all $0\leq t'\leq
t$ that $(i,t',a_{t'})\in\responseSubset$.

For a strategy profile $\stratp$ denote the optimal expected utility
for $i$ under any response as $\util_i^\star(\stratp)=\sup_{\stratpx}
\util_i(\stratpx)$ where the supremum is over strategy profiles
$\stratpx$ such that $\stratx^j=\strat^j$ for all $j \neq i$ in $V$.

\begin{definition}
  On a history-closed subset $\responseSubset\in\responseSpace$ a \emph{forced
    response} $q_\responseSubset$ is a map
  $q_\responseSubset:\responseSubset\to \{0,1\}$ denoting a set of
  actions we force the agents to make.  A strategy profile $\stratp$
  is \emph{$q_\responseSubset$-forced} if for every
  $(i,t,a)\in\responseSubset$ if agent $i$ at time $t$ has seen
  history $a$ from her neighbors then she selects action
  $q_\responseSubset(i,t,a)$.  A strategy profile $\stratp$ is a
  \emph{$q_\responseSubset$-equilibrium} if it is
  $q_\responseSubset$-forced and for every agent $i \in V$ it holds
  that $\util_i(\stratp) \geq \util_i(\stratpx)$ for any
  $q_\responseSubset$-forced strategy profile $\stratpx$ such that
  $\stratx^j=\strat^j$ for all $j \neq i$ in $V$.

%We say that a strategy $\stratpx$ is a \emph{$(i,t,a)$-adjustment} of $\stratp$ if the action of agents $i$ at time $t$ and history $a$ is changed and the strategies only differs in the actions of agent $i$ at times $t'\geq t$ and with histories $a'$ such that $a'_t=a$.  The improvement of the adjustment is denoted by
%\[
%\psi(\stratp,i,t,a) = \sup_\stratpx (1-\disc)\sum_{t = T}^\infty\disc^t\left(\util_{i,t}(G,\stratpx) - \util_{i,t}(G,\stratp)\right),
%\]
%where the supremum is over all $(i,t,a)$-adjustments of $\stratp$.
\end{definition}

The following lemma can be proved by a minor modification of Theorem~\ref{thm:equi-exists} and so we omit the proof.

\begin{lemma}\label{l:forcedEqui-exists}
Let $\responseSubset\in\responseSpace$ be history-closed and let  $q_\responseSubset$ be a forced response.  There exists a $q_\responseSubset$-equilibrium.
\end{lemma}

Having constructed $q_\responseSubset$-equilibria we then will want to show that they are equilibria.    In order to do that we appeal to the following lemma.

\begin{lemma}\label{l:fullEquil}
  Let $\stratp$ be a $q_\responseSubset$-equilibrium.  Suppose that
  for every agent $i$, any strategy profile $\stratpx$ that attains
  $\util_i^\star(\stratp)$ has that for all $t$,
\begin{equation}\label{e:sameAction}
\P{\stratp_t^i(\belief_i, \action^{\neigh{i}}_{[0, t)}) \neq \stratpx_t^i(\belief_i, \action^{\neigh{i}}_{[0, t)}),(i,t,(\belief_i, \action^{\neigh{i}}_{[0, t)}))\in\responseSubset} = 0.
\end{equation}
Then $\stratp$ is an equilibrium.
\end{lemma}

\begin{proof}
  If $\stratp$ is not an equilibrium then by compactness there exists
  a strategy profile for $\stratpx$ that attains $\util_i^\star$ and
  differs from $\stratp$ only for agent $i$.  By
  equation~\eqref{e:sameAction} this implies that agent $i$ following
  $\stratpx$ must take the same actions almost surely as if they were
  following $\stratp$ until the end of the forced moves.  Hence it is
  $q_\responseSubset$-forced and so $\stratpx$ is a
  $q_\responseSubset$-equilibrium.  It follows that $i$ cannot
  increase the expected utility of $\stratp$, which is therefore an
  equilibrium.
\end{proof}

In order to show that every agent follows the forced moves almost
surely we now give a lemma which gives a sufficient condition for an
agent to act myopically, according to her posterior distribution.  For
an equilibrium strategy profile $\stratp$ let
$\stratp^\dagger_{i,t,a}$ be the strategy profile where the agents
follow $\stratp$ except that if agent $i$ has $a=(\belief_i,
\action^{\neigh{i}}_{[0, t)})$
then from time $t$ onwards agent $i$ acts myopically, taking action
$\bestres^i_{t'}(G,\stratp^\dagger_{i,t,a})$ for time $t'\geq t$.  We denote
\[
Y_\ell=Y_\ell(i,t,a):=\mathbb{E}\Bigg[\left|\CondP{S=1}{\cF^i_{t+\ell}(G,\stratp^\dagger_{i,t,a})}
    - 1/2\right|\;\;\Bigg\vert\;\cF^i_t,a=(\belief_i, \action^{\neigh{i}}_{[0, t)})\Bigg].
\]
We will show that the following are sufficient conditions for agent
$i$ to act myopically.  For $\ell\in\{1,2,3\}$ we set
$\cD_{\ell}=\left\{2Y_0 > \frac{\disc^2 (\tfrac12 -Y_{\ell-1})
  }{1-\disc}\right\}$ and we set
\[
\cD_4=\left\{2Y_0 > \disc^2 (\tfrac12 -Y_2) + \frac{\disc^3 (\tfrac12 -Y_3) }{1-\disc}\right\}.
\]
Since $\stratp$ and $\stratp^\dagger_{i,t,a}$ are the same up to time $t-1$ we have that $\cF^i_t(G,\stratp)$ is equal to $\cF^i_t(G,\stratp^\dagger_{i,t,a})$.
As $Y_\ell$ is the expectation of a submartingale it is increasing. Hence, %noting that $Y_0=\left|\CondP{S=1}{a=(\belief_i, \action^{\neigh{i}}_{[0, t)})} - \tfrac12\right|$,
after rearranging we see that $\cD_1\subseteq \cD_2\subseteq \cD_3 \subseteq \cD_4$.

\begin{lemma}\label{l:myopicCond}
Suppose that for strategy profile $\stratp$ agent $i$ has an optimal response, such that for any $\stratpx$ such that $\stratx^j=\strat^j$ for all $j \neq i$ in $V$ then $\util_i(\stratp) \geq \util_i(\stratpx)$.
Then for any $t$,
\[
\P{\action^i_t(G,\stratp) \neq \bestres^i_t, \cD_1\cup \cD_2 \cup \cD_3 \cup \cD_4}=0,
\]
that is, agent $i$ acts myopically at time $t$ under $\stratp$ almost surely, on the event $\cD_1\cup \cD_2 \cup \cD_3 \cup \cD_4$.
\end{lemma}

\begin{proof}
  If agent $i$ acts under $\stratp^\dagger_{i,t,a}$ then her expected
  utility from time $t$ onwards given $a$ is
\begin{align*}
\util_{i,t,a}(\stratp^\dagger_{i,t,a}) &:=
(1-\disc)\sum_{t'=t}^\infty\disc^{t'}
\CondE{\P{\action^i_{t'}(G,\stratp^\dagger_{i,t}) =
    S}}{\cF^i_t,a=(\belief_i, \action^{\neigh{i}}_{[0, t)})}\\
&\geq (1-\disc)\disc^t\bigg(\tfrac12+ Y_0  + \disc \left(\tfrac12+ Y_1\right) + \disc^2 \left(\tfrac12+ Y_2\right) + \frac{\disc^3}{1-\disc}\left(\tfrac12+ Y_3 \right) \bigg)
\end{align*}
under $\stratp^\dagger_{i,t,a}$.  Now assume that the action of agent
$i$ at time $t$ under $\stratp$ is not the myopic choice. Then her
expected utility is at most
\begin{align*}
\util_{i,t,a}(\stratp) &\leq (1-\disc)\disc^t\bigg(\frac12 -
\left|\CondP{S=1}{\cF^i_t,a=(\belief_i, \action^{\neigh{i}}_{[0, t)})} - \tfrac12\right| \\
&\qquad + \disc \CondE{\P{\action^i_{t+1}(G,\stratp) = S}}{\cF^i_t,a=(\belief_i, \action^{\neigh{i}}_{[0, t)})} +\frac{\disc^2}{1-\disc} \bigg).
\end{align*}

We note that at time $t+1$ the information available about $S$ is the same under both strategies since the only difference  is the choice of action by agent $i$ at time $t$, hence as $i$ takes the optimal action under $\stratp^\dagger$,
\begin{align*}
\tfrac12+ Y_1=\CondE{\P{\action^i_{t+1}(G,\stratp^\dagger_{i,t,a}) = S}}{\cF^i_t,a=(\belief_i, \action^{\neigh{i}}_{[0, t)})} \geq \CondE{\P{\action^i_{t+1}(G,\stratp) = S}}{\cF^i_t,a=(\belief_i, \action^{\neigh{i}}_{[0, t)})}.
\end{align*}
Since $\stratp$ is optimal for $i$ we have that
\begin{align}\label{e:utilComparison}
0\geq\util_{i,t,a}(\stratp^\dagger_{i,t,a}) - \util_{i,t}(\stratp)
\geq (1-\disc)\disc^t\left( 2 Y_0 - \disc^2 \left(\tfrac12-Y_2\right) - \frac{\disc^3}{1-\disc} \left(\tfrac12-Y_3\right)\right).
\end{align}
Condition~\eqref{e:utilComparison} does not hold under $\cD_4$ so $\P{\action^i_t(G,\stratp) \neq \bestres^i_t,\cD_1\cup \cD_2 \cup \cD_3 \cup \cD_4}=0$.
\end{proof}

%We note that for the conclusion we don't need to assume that of the lemma it is sufficient to agent $i$ is in {\em equilibrium}, that is for any $\stratpx$ such that $\stratx^j=\strat^j$ for all $j \neq i$ in $V$ then $\util_i(\stratp) \geq \util_i(\stratpx)$.

\subsection{The royal family}

In the main theorem we require that the graph $G$ not only be strongly
connected, but also $L$-connected and have bounded out-degrees, which
are local conditions.  In the following example the graph is strongly
connected, has bounded out-degrees, but is not $L$-connected. We show
that for bounded private beliefs asymptotic learning does not occur in
all equilibria\footnote{We draw on Bala and
  Goyal's~\cite{BalaGoyal:96} {\em royal family} graph.}.

\begin{figure}[h]
  \centering
  \includegraphics[scale=0.8]{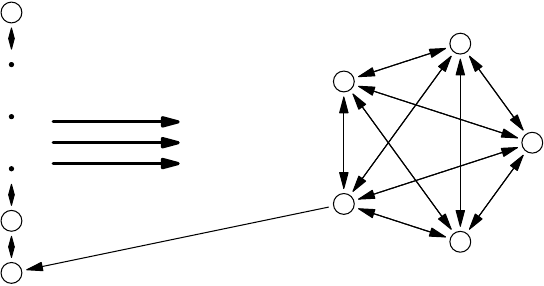}
  \caption{\label{fig:royal-family-app} The Royal Family. Each member of
    the public (on the left), observes each royal (on the right), as
    well as her next door neighbors. The royals observe each other,
    and one royal observes one member of the public. }
\end{figure}

Consider the following graph (Figure~\ref{fig:royal-family-app}). The
vertex set is comprised of two groups of agents: a ``royal family''
clique of $R$ agents who all observe each other, and $n \in \N \cup
\{\infty\}$ agents - the ``public'' - who are connected in an
undirected chain, and in addition can all observe all the agents in
the royal family. Finally, a single member of the royal family
observes one of the public, so that the graph is strongly connected.

We choose that $\mu_0$ and $\mu_1$ so that $\P{Z^i_0 \in (1,2)\cup
  (-2,-1)}=1$ and set the forced moves so that all agents act
myopically at time 1. By Lemma~\ref{l:forcedEqui-exists} we can extend
this to a forced equilibrium $\stratp$.  By Lemma~\ref{l:fullEquil} it
is sufficient to show that no agent can achieve their optimum without
choosing the myopic action in the first round.  By our choice of
$\mu_0$ and $\mu_1$ we have that
\[
\left|\CondP{S=1}{\cF^i_0}-\frac12\right|=\frac{e^{|Z^i_0|}}{1+e^{|Z^i_0|}}-\frac12 \geq \frac{e}{1+e}-\frac12 \geq \frac15.
\]
Hence in the notation of Lemma~\ref{l:myopicCond} we have that $Y_0\geq \frac15$ when $t=0$ for all $i$ and $a$ almost surely.  Moreover, after the first round all agents see the royal family and can combine their information.  Since the signals are bounded it follows that for some $c=c(\mu_0,\mu_1)>0$, independent or $R$ and $n$
\[
\CondE{\frac12 - \left|\CondP{S=1}{\cF^i_1}-\frac12\right|}{\cF^i_0} \leq e^{-c R}.
\]
Hence if $R$ is a large constant $\cD_2$ holds so by
Lemma~\ref{l:myopicCond} if an agent is to attain her maximal expected
utility given the actions of the other agents, she must act myopically
almost surely at time 0.  Thus $\stratp$ is an equilibrium.

Let $\cJ$ denote the event that all agents in the royal family have a signal favoring state 1.  On this event under $\stratp$ all agents in the royal family choose action 1 at time 0 and this is observed by all the agents so $\cJ\in\cF^i_1$ for all $i$.  Since agents observe at most one other agent this signal overwhelms their other information and so
\[
\CondP{S=1}{\cF^i_1,\cJ} \geq 1- e^{-c R},
\]
for all $i\in V$.  Thus if $R$ is a large constant $\cD_1$ holds for
all the agents at time 1 so by Lemma~\ref{l:myopicCond} they all act
myopically and choose action 1 at time 1.  Since $\cJ\in\cF^i_1$ they
also all knew this was what would happen so gain no extra information.
Iterating this argument we see that all agents choose 1 in all
subsequent rounds.  However, $\P{\cJ,S=0}\geq e^{-c' R}$ where $c'$ is
independent of $R$ and $n$.  Hence as we let $n$ tend to infinity the
probability of learning does not tend to 1, and when $n$ equals
infinity the probability of learning does not equal 1.

\subsection{The mad king}

More surprising is that there exist \emph{undirected} (i.e.,
$1$-connected) graphs with equilibria where asymptotic learning fails;
These graphs have unbounded out-degrees. Note that in the myopic case
learning is achieved on these graphs~\cite{mossel2012asymptotic}, and
so this is an example in which strategic behavior impedes learning.

In this example we consider a finite graph which includes 5 classes of
agents.  There is a king, labeled $u$, and a regent labeled $v$.  The
court consists of $R_C$ agents and the bureaucracy of $R_B$ agents.  The
remaining $n$ are the people. Note again that the graph is
undirected.
\begin{itemize}
\item The king  is connected to the regent, the court and the people.
\item The regent is connected to the king and to the bureaucracy.
\item The members of the court are each connected only to the king.
\item The members of the people are each connected only to the king.
\item The members of the bureaucracy are each connected only to the regent.
\end{itemize}
See Figure~\ref{fig:mad-royal-family-app}.
\begin{figure}[h]
  \centering
  \includegraphics[scale=1]{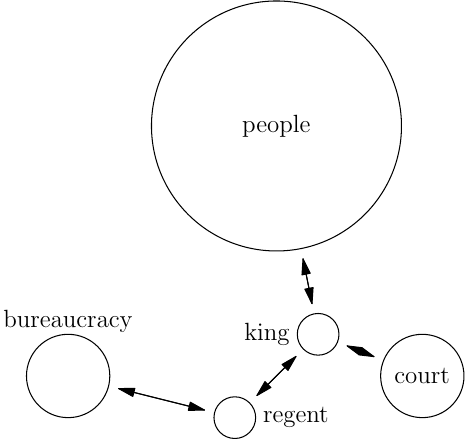}
  \caption{\label{fig:mad-royal-family-app} The mad king.}
\end{figure}

As in the previous example we will describe some initial forced equilibrium
and then appeal to existence results to extend it to an equilibrium.
We suppose that $\mu_0$ and $\mu_1$ are such that $\P{ Z^i_0 \in
  (1,1+\epsilon)\cup (-\sqrt{7},-\sqrt{7}+\epsilon)}=1$ where
$\epsilon$ is some very small positive constant, and will choose
$R_C,\disc$ and $R_B$ so that $e^{R_C}$ is much smaller than
$\frac1{1-\disc}$ which in turn will be much smaller than $ R_B$:
\begin{align*}
  e^{R_C} \ll \frac1{1-\disc} \ll R_B.
\end{align*}
The equilibrium we describe will involve the people being forced to
choose action 0 in rounds 0 and 1, as otherwise the king ``punishes''
them by withholding his information. As an incentive to comply he
offers them the opinion of his court and, later, of his
bureaucracy. While the opinion of the bureaucracy is correct with high
probability, it is still bounded, and so, even as the size of the
public tends to infinity, the probability of learning stays bounded
away from one.

We now describe a series of forced moves for the agents, fixing
$\delta>0$ to be some small constant.

\begin{itemize}
\item The regent acts myopically at time 0.  If for some state $s$
  $\CondP{S=s}{\cF^{v}_1} \geq 1 - e^{-\delta R_B}$ then the regent
  chooses states $s$ in round 1 and all future rounds, otherwise his
  moves are not forced.
\item The king acts myopically in rounds 0 and 1 unless one or more of
  the people chooses action 1 in round 0 or 1, in which case he chooses
  action 1 in all future rounds.  Otherwise if $s$ is the action of
  the regent at time 1 then from time 2 the king takes action $s$
  until the regent deviates and chooses another action.
\item The members of the bureaucracy act myopically in round 0 and 1.
  If $s$ is the action of the regent at time 1 then from time 2 the
  members of the bureaucracy take action $s$ until the regent deviates
  and chooses another action.
\item The members of the court act myopically in round 0 and 1.  At
  time 2 they copy the action of the king from time 1.  If $s$ is the
  action of the king at time 2 then from time 3 the members of the
  bureaucracy take action $s$ until the king deviates and chooses
  another action.
\item The people choose action 0 in rounds 1 and 2.  At time 2 they
  copy the action of the king from time 1.  If $s$ is the action of
  the king at time 2 then from time 3 the people take action $s$ until
  the king deviates and chooses another action.
\end{itemize}

By Lemma~\ref{l:forcedEqui-exists} this can be extended to a forced
equilibrium strategy $\stratp$.  We will show that this is also an
equilibrium strategy in the unrestricted game by establishing
equation~\eqref{e:sameAction}.  In what follows when we say acts optimally or
in an optimal strategy we mean for an agent with respect to the
actions of the other agents under $\stratp$.

First consider the regent.  By our choice of $\mu_0,\mu_1$ we have
that $Y_0>\frac15$.  Let $\cJ=\cJ_0\cup \cJ_1$ where $\cJ_s$ denotes
the event that $\CondP{S=s}{\cF^{v}_1} \geq 1 - e^{-\delta R_B}$.
Since the regent views all the myopic actions of the bureaucracy he
knows the correct value of $S$ except with probability exponentially
small in $R_B$ so for $s\in\{0,1\}$, if $\delta>0$ is small enough,
\[
\CondP{\cJ_s}{S=s} \geq 1 - e^{-\delta R_B}
\]
and hence for large enough $R_B$ we have that $Y_1 \geq \frac12 -
2e^{-\delta R_B}$ which implies that $\cD_2$ holds at time 1.  By
Lemma~\ref{l:myopicCond} in any optimal strategy the regent acts
myopically in round 0, and so follows the forced move.  On the event
$\cJ_s$ the regent follows $s$ in all future steps.  At time 1
condition $\cD_1$ holds so again the regent follows the forced move in
any optimal strategy.  We next claim that for large enough $R_B$
\begin{equation}\label{e:regentKnowledge}
\CondP{\CondP{S=s}{\cF^{v}_2}\geq 1 - e^{-\delta R_B/2}}{\cJ_s} =1
\end{equation}
Assuming \eqref{e:regentKnowledge} holds then condition $\cD_1$ again
holds so the regent must choose $s$ at time 2 in any optimal strategy.
By construction of the forced moves from time 2 onwards the king and
bureaucracy simply imitate the regent and so he receives no further
information from time 2 onwards.  Thus again using
Lemma~\ref{l:myopicCond} we see that under any optimal strategy the
regent must follow his forced moves.

To establish that the regent follows the forced moves in any optimal
strategy it remains to show that Condition~\eqref{e:regentKnowledge}
holds.  The information available to the regent at time 2 includes the
actions of the king and the bureaucracy at times 0 and 1.  Consider
the actions of the bureaucracy at times 0 and 1.  At time 0 they
follows their initial signal.  At time 1 they also learn the initial
action of the regent who acts myopically. By our assumption on $\mu_0$
and $\mu_1$ that $\P{Z^i_0 \in (1,1+\epsilon)\cup
  (-\sqrt{7},-\sqrt{7}+\epsilon)}=1$, an initial signal towards 0 is
much stronger than an initial signal towards 1, since whenever $Z$ is
negative it is at most $-\sqrt{7}+\epsilon$.  For $i$, a member of the
bureaucracy, we have that $Z_1^i \geq 2$ if both $i$ and the regent
choose action 1 at time 1.  However, if either $i$ or the regent
choose action 0 at time 1 then $Z^i_t \leq -\sqrt{7}+\epsilon + 1 +
\epsilon < -1$.  Since the actions of $i$ and the regent at time 0 are
known to the regent at time 1, he gains no extra information at time 2
from his observation of $i$ at time 1 since he can correctly predict
his action.

The information the regent has available at time 2 is thus his information from time 1 together with the information from observing the king.  The information available to the king is a function of his initial signal and that of the regent and the court.  Since this is only $R_C+1$ members and we choose $R_B$ to be much larger than $R_C$ it is insignificant compared to the information the regent observed from the court at time 0 and hence \eqref{e:regentKnowledge} holds.  Thus, there is no optimal strategy for the regent that deviates from the forced moves.

As we noted above the members of the bureaucracy have $|Z_0^i|,|Z_1^i|  \geq 1$ almost surely.  For $t \geq 1$ let $\cM_{s,t}$ denote the event that the regent chose action $s$ for times 1 up to $t$.  As argued above, $\cJ_s \subset \cM_{s,t}$ for all $t$ under $\stratp$.  This analysis holds even if a single member of the bureaucracy adopts a different strategy as we have taken $R_B$ to be large so this change is insignificant.  Given that  $\cM_{s,t}$ holds, the only additional information available to agent $i$, a member of the bureaucracy, is their original signal and the action at time 1 of the regent.  Thus
\[
\CondP{S=s}{\cF^i_t,\cM_{s,t}} \geq 1 - e^{-\delta R_B/2}.
\]
It follows then by Lemma~\ref{l:myopicCond} that acting myopically at
times 0 and 1 and then imitating the regent until he changes his
action is the sole optimal strategy for a member of the bureaucracy.

Next consider the forced responses of the king.  Since under $\stratp$ the people always choose action 0 at times 0 and 1, the rule forcing the king to choose action 1 after seeing a 1 from the people is never invoked.  We claim that, provided $R_B$ is taken to be sufficiently large, that the king acts myopically at times 0 and 1.  At time 0 the posterior probability of $S=1$ is bounded away from 1/2 so $Y_0$ is bounded away from 0 while $\frac12-Y_2 \leq 2e^{-\delta R_B/2}$ so by Lemma~\ref{l:myopicCond} the king must act myopically.  Similarly at time 1 since our choice of $\mu_0$ and $\mu_1$ to have their log-likelihood ratio concentrated around either 1 or $-\sqrt{7}$ a posterior  calculation gives that,
\[
|Z_1^u - \#\{i\in \neigh{u} :\action_0^i(\stratp)=1\} +
\sqrt{7}\#\{i\in \neigh{u} :\action_0^i=0\}| \leq \epsilon (2+R_C)
\]
and thus for some $\epsilon(R_C)>0$ sufficiently small we can find an
$\epsilon'(\epsilon,R_C)>0$ such that
$Y_0=|\frac{e^{Z_1^u}}{1+e^{Z_1^u}}-\frac12|> \epsilon'$.  Since we
again have that $\frac12-Y_1 \leq 2e^{-\delta R_B/2}$ taking
$R_B=R_B(\epsilon,R_C)$ to be sufficiently large $\cB_2$ holds and so
the king must act myopically.  It remains to see that the king should
imitate the regent from time 2 onwards unless the regent subsequently
changes his action in any optimal strategy.  This follows from a
similar analysis to the case of the members of the bureaucracy so we
omit it.

We next move to an agent $i$, a member of the court.  At time 0 the agent has $Y_0>\frac{e}{1+e}-\frac12>\frac15$.  Agent $i$ at time 1 views the action of the king who has in turn viewed the actions of the whole court at time 0 so $\frac12-Y_2 \leq e^{-c R_C}$. At time 2 the agent sees the action of the king who has imitated the action of the regent at time 1 so $\frac12-Y_3 \leq e^{-\delta R_B/2}$.  Hence provided that $R_C$ is sufficiently large and $R_B(R_C,\disc)$ is sufficiently large then $\cB_4$ holds and $i$ must act myopically at time 0.  The information of a member of the court at time 1 is a combination of their initial signal and the action of the king at time 1.  Similarly to a member of the bureaucracy, by the choice of $\mu_0$ and $\mu_1$ we have that $|Z_1^i| \geq 1$ and so $Y_0>\frac15$.  Also $\frac12-Y_2 \leq e^{-\delta R_B/2}$ since this includes the information from the action of the regent at time 1.  Thus $\cB_3$ holds and $i$ must act myopically at time 1.  At time 2 agents $i$ knows the action of the king from round 2 so $Y_0 \geq \frac12 - e^{-c R_C}$ and $\frac12-Y_1 \leq e^{-\delta R_B/2}$ so $\cB_2$ holds and $i$ must act myopically at time 2.  Finally from time 3 onwards agent $i$ knows the action of the regent at time 1.  As with the king and bureaucracy this will not be changed unless $i$ receives new information, that is the king changes his action sometime after time 2.  Thus any optimal strategy of $i$ follows the forced moves.

This finally leaves the people.  Let agent $i$ be one of the
people. We first check that it is always better for them to wait and
just say 0 in rounds 0 and 1 in order to get more information from the
king, their only source.  If agent $i$ chooses action 1 at time 0 then
the total information it receives is a function of the initial signals
of $i$ and the king.  Thus, since the signals are uniformly bounded,
even if the agent knew the signals exactly we would have that for some
$c'(\mu_0,\mu_1)$ that the expected utility from such a strategy is at
most $1-e^{-2c'}$.  If an agent acts with 0 at time 0 but 1 at time
1, she can potentially receive information from the initial signals of
the king, court and regent as well as her own; still, the optimal
expected utility even using all of this information is at most
$1-e^{-c'(R_C+3)}$.  Consider instead the expected utility following
the forced moves. On the event $\cJ$ agent $i$ will have expected
utility at least $\disc^3 (1-e^{-\delta R_B})$ which is greater than
$1-e^{-c'(R_C+3)}$ provided that $\disc$ is sufficiently close to 1
and $R_B$ is sufficiently large.  Thus agent $i$ must choose action 0
at times 0 and 1 in any optimal strategy.  The analysis of rounds 2
and onwards follows similarly to the court and thus any optimal
strategy of $i$ follows all the forced moves.

This exhaustively shows that there is no alternative optimal strategy for any of the agents which differs from the forced moves.  Thus $\stratp$ is an equilibrium.  However, on the event $\cJ_1$ all the agents actions converge to 1. However,  $\P{\cJ,S=0}\geq e^{-c'' R_B}>0$ where $c''$ is independent of $R_C,R_B,\disc$ and $n$.  Hence, as we let $n$ tend to infinity the probability of learning does not tend to 1.

\end{document}